\documentclass[11pt]{article}
\usepackage{amssymb}
\usepackage{amsfonts}
\usepackage{amsthm}
\usepackage{graphicx}
\usepackage[utf8]{inputenc}
\usepackage{xcolor}
\usepackage[english]{babel}
\usepackage{amsmath}
\usepackage{eucal}
\usepackage{enumerate}
\usepackage{booktabs}
\usepackage{stmaryrd}
\usepackage{amsfonts,epsfig,epstopdf, titling,url, array}
\usepackage{verbatim} 
\usepackage{hyperref}
\usepackage{subcaption}
\usepackage{algorithm}
\usepackage{algpseudocode}

\newcommand\bfY{\mathbf{Y}}
\newcommand\bff{\mathbf{f}}
\newcommand\bfx{\mathbf{x}}
\newcommand\bfa{\mathbf{a}}
\newcommand\bfb{\mathbf{b}}
\newcommand\bfv{\mathbf{v}}
\newcommand\bfw{\mathbf{w}}
\newcommand\bfI{\mathbf{I}}

\numberwithin{equation}{section}

\newtheorem{thm}{Theorem}

\newcommand{\beq}{\begin{equation}}
\newcommand{\RN}[1]{%
  \textup{\uppercase\expandafter{\romannumeral#1}}%
}
\newcommand{\eeq}{\end{equation}}

\newcommand{\beqs}{\begin{equation*}}
\newcommand{\eeqs}{\end{equation*}}

\usepackage[
top    = 2.75cm,
bottom = 2.50cm,
left   = 2.49cm,
right  = 2.549cm]{geometry}

\begin{document}
\begin{titlepage}
\LARGE{
\begin{center}
The Curse of Black Sigatoka: A Backward Bifurcation Perspective
\end{center}
}
\large{
\begin{center} 
Bernard Asamoah Afful\footnote{Corresponding author, e-mail address: bernard.afful@usu.edu}, Luis F. Gordillo\footnote{e-mail address:  luis.gordillo@usu.edu} \\[0.3cm]
Department of Mathematics and Statistics, Utah State University, Logan, UT 84322\\
\end{center}
\begin{abstract}
\noindent 
Black Sigatoka disease (BSD), also known as black leaf streak disease, is an airborne fungal infection caused by \textit{Pseudocercospora fijiensis} that severely impacts global banana and plantain production. Its persistence and resistance to eradication make it one of the most challenging plant diseases to manage. In this paper, we propose a deterministic pathogen-host model to describe BSD dynamics. Due to dual transmission pathways (ascospores and conidia) and mate limitation in sexual reproduction, the model exhibits a backward bifurcation: a stable endemic equilibrium coexists with the disease-free equilibrium for certain parameter values in which the basic reproduction number, $\mathcal{R}_0$, is less than 1. This phenomenon explains why control strategies that solely reduce $\mathcal{R}_0$ below one may fail. For the backward bifurcation regime, we perform sensitivity analysis of the endemic equilibrium using normalized forward sensitivity indices, Latin Hypercube Sampling, and Partial Rank Correlation Coefficients. Results indicate that effective control must extend beyond $\mathcal{R}_0$ reduction and prioritize (1) limiting production of new susceptible leaves during high-risk periods and (2) developing and deploying disease-resistant plant varieties. To incorporate transmission variability, we also formulate a stochastic version of the model using the Stochastic Simulation Algorithm (SSA). Extensive numerical simulations compare stochastic realizations with deterministic predictions and quantify variability in disease dynamics. To identify the principal drivers of persistence and variability, we analyze the endemic equilibrium using Sobol's variance-based sensitivity method, which highlights the role of nonlinear parameter interactions in shaping variability.
\\[0.3cm]
\end{abstract}
}
\end{titlepage}

\section{Introduction}\label{sec:intro}
Black Sigatoka disease (BSD), also known as black leaf streak disease (BLSD), is an airborne fungal disease caused by \textit{Pseudocercospora fijiensis} (formerly \textit{Mycosphaerella fijiensis}). It is endemic in most banana- and plantain-producing regions worldwide and ranks among the most severe threats to global banana production. In susceptible cultivars, the pathogen induces extensive leaf necrosis and premature defoliation before fruit bunches reach maturity, resulting in yield losses that frequently exceed 50\% in the absence of control measures.

Despite intensive management efforts, BSD remains exceptionally difficult to eradicate or suppress effectively. Several factors contribute to this resilience: the pathogen produces abundant ascospores that facilitate long-distance wind dispersal and genetic recombination; resistance to multiple fungicide classes has emerged following repeated applications; and the fungus is well adapted to warm, humid tropical environments that promote short disease cycles and frequent reinfection. These characteristics make containment or complete eradication nearly impossible in most production systems. Reliance on frequent fungicide sprays achieves only partial suppression, at considerable economic and environmental cost.

A rare exception occurred in North Queensland, Australia. Following the detection of BSD in April 2001 in the Tully region of Australia, the country's largest commercial banana-growing area, an intensive, multi-year eradication campaign was implemented. This program involved rigorous surveillance, destruction of infected host material where necessary, and strict quarantine measures. The effort successfully eliminated the pathogen from commercial plantations, representing perhaps the only documented instance of BSD eradication from a major production zone.

BSD spreads via two spore types, asexual conidia and sexual ascospores, which play complementary roles in the epidemic cycle. Conidia, produced in young to intermediate lesions, drive rapid local spread through rain splash, water runoff, or localized wind, sustaining polycyclic progression within plantations. Ascospores, maturing on necrotic tissue and forcibly discharged under wet conditions, enable broader geographic dispersal (often tens to hundreds of kilometers) and serve as primary inoculum for initiating epidemics in new areas or reinfecting plantations after seasonal lows or sanitation efforts. The combination of these dispersal mechanisms creates a highly effective transmission system.

In this paper, we present a quantitative compartmental pathogen-host model using differential equations to describe BSD dynamics. The model incorporates both spore types and accounts for mate limitation in ascospore production (due to the nature of \textit{P. fijiensis}, which requires compatible mating types). We derive the basic reproduction number, $\mathcal{R}_0$, and demonstrate that the model exhibits a backward bifurcation, by which a stable endemic equilibrium coexists with the disease-free equilibrium for $\mathcal{R}_0<1$. This bistability provides a theoretical explanation for the disease's extraordinary resilience and the failure of control strategies that aim solely to reduce $\mathcal{R}_0$ below 1.
We also formulate a stochastic model, implementing the Stochastic Simulation Algorithm (SSA) \cite{gillespie}, to investigate how intrinsic, individual randomness affects the overall variability of the disease dynamics.

The paper is structured as follows. In Sections 2 and 5, we describe the deterministic and stochastic compartmental models, respectively. Section 3 derives the disease-free and endemic equilibria and computes $\mathcal{R}_0$, and the analysis of backward bifurcation. In Section 4, we conduct sensitivity analysis of the endemic equilibrium using normalized forward sensitivity indices, Latin Hypercube Sampling, and Partial Rank Correlation Coefficients. Section 5 presents comprehensive numerical simulations that compare deterministic predictions with stochastic realizations, examine persistence/eradication thresholds under varying environmental and biological conditions, assess eradication probabilities at small population sizes, and use Sobol's variance-based method for quantifying the contributions of individual parameters and their interactions to the variability of the endemic equilibrium. Finally, the conclusions are presented in Section 6.

\section{A deterministic model}\label{sec:ode}
The life cycle of the causal agent of Black Sigatoka disease (BSD), \textit{Pseudocercospora fijiensis} (formerly \textit{Mycosphaerella fijiensis}), is illustrated in Figure \eqref{cycle}. The fungus completes its entire infection cycle on the banana plant, progressing through four distinct phases: (1) exposure, (2) incubation, (3) latency, and (4) spore production. As with many foliar fungal pathogens, spore germination and infection require high relative humidity or a wet leaf surface. Consequently, relative humidity and temperature strongly influence disease progression. In BSD, leaf susceptibility decreases with age, making young leaves the primary sites for spore deposition and germination. Following germination, lesions develop through exposure and latency phases before producing asexual conidia and, upon maturation of necrotic tissue, sexual ascospores. Spore germination typically occurs within 2 to 3 hours under conditions of high humidity and free water availability during the exposure phase.
\begin{figure}[!ht]
    \centering
     \includegraphics[scale=0.3]{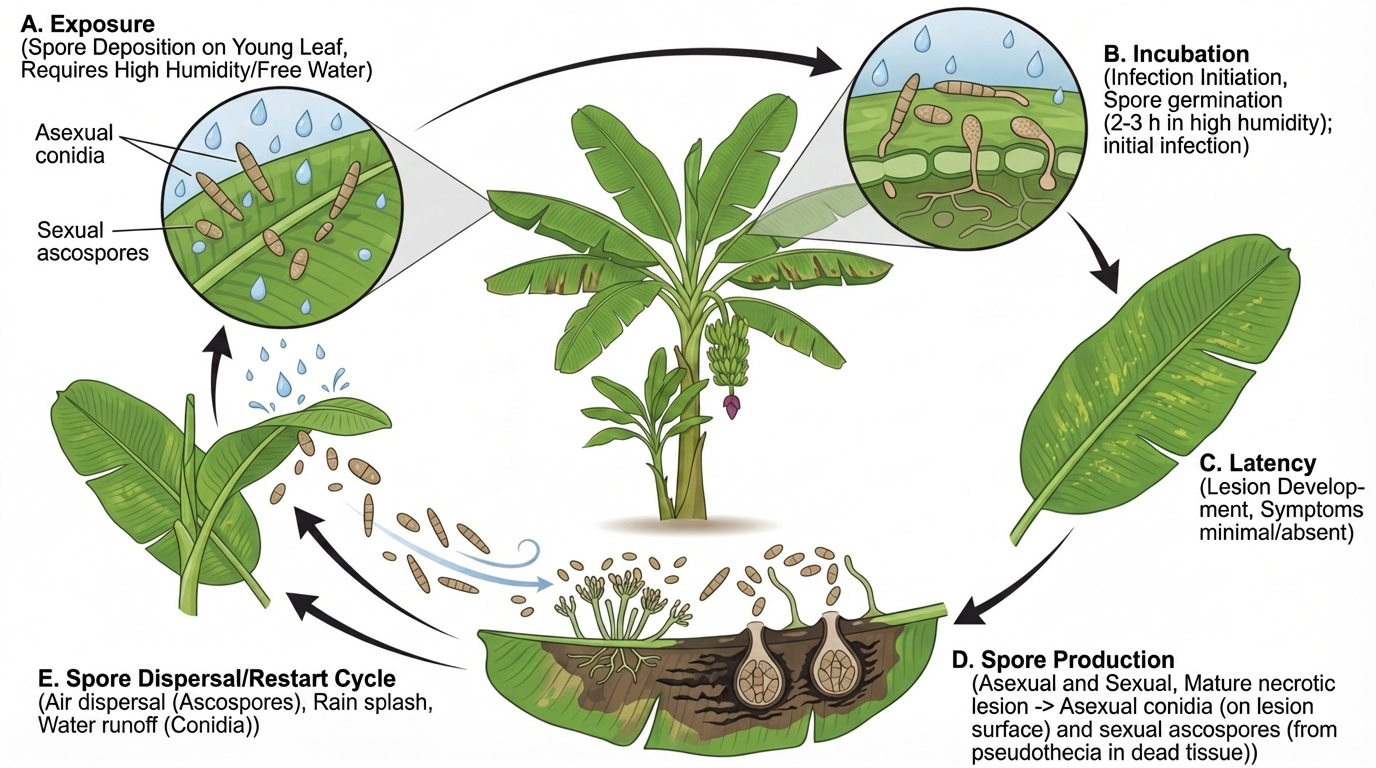}
    \caption{Life cycle of Black Sigatoka disease (BSD), caused by \textit{Pseudocercospora fijiensis} on a banana plant. The diagram illustrates the progression from spore deposition on young leaves through infection, incubation, latency, asexual spore (conidia) production in developing lesions, and eventual sexual maturation on necrotic tissue, leading to ascospore production in pseudothecia. Both conidia and ascospores serve as inoculum for new infections. Arrows indicate key transitions and spore dispersal pathways. (Adapted from \cite{ravigne} using \url{figurelabs.ai}).}
    \label{cycle}
\end{figure}

Let $N(t)$ denote the total leaf biomass density at time $t$. This biomass is partitioned into four compartments: susceptible $H(t)$, resistant $R(t)$, exposed $E(t)$, and infected $I(t)$, that is, $N(t)= H(t) + R(t) + E(t) + I(t)$. Healthy leaf biomass consists of susceptible and resistant biomass, $H(t) + R(t)$. Susceptible leaves permit rapid spore germination and infection under favorable climatic conditions, whereas resistant leaves inhibit initial spore development, even in suitable environments; spores on resistant leaves may fail to germinate or do so slowly.
New leaf biomass enters the system at a rate $\Lambda$. The fraction of susceptible leaves is denoted by $\kappa$ (with $1 - \kappa$ resistant). The natural decay rate of leaves is assumed constant and denoted by $\mu$.

\subsection{The transmission rate}
BSD is transmitted via two complementary spore types: asexual conidia and sexual ascospores. The pathogen exhibits a polycyclic life cycle, with multiple infection generations possible within a single growing season. Both spore types initiate infections on susceptible leaves, producing new lesions and perpetuating transmission.

Conidia, produced in young to intermediate lesions, are dispersed primarily over short distances by rain splash, water runoff, or localized wind. They drive rapid local epidemics and sustain polycyclic progression within plantations. Conidia germinate under high humidity and free water on leaf surfaces, penetrating stomata to form new lesions. Although fewer conidia are produced per lesion than ascospores, they are essential for building local epidemic intensity.

Ascospores mature on necrotic leaf tissue and are forcibly discharged under wet conditions. Smaller and more abundant per lesion, ascospores are primarily wind-dispersed, enabling long-distance movement (tens to hundreds of kilometers, though limited by UV sensitivity). They serve as the primary inoculum for initiating epidemics in new areas or reinfecting plantations after sanitation or seasonal declines, while also contributing to within-plantation spread and genetic recombination (e.g., facilitating fungicide resistance adaptation).

The densities of conidia and ascospores at time $t$ are denoted by $A(t)$ and $S(t)$, respectively. The transmission rate $\beta$ depends on relative humidity $h$ and temperature $T$, 
\begin{equation}
    \beta(h,T)=\beta_0\overbrace{\left(\frac{h}{h+K_h}\right)}^{\text{Humidity effects}}\underbrace{\left(\exp{\left(-\frac{(T-\hat{T})^2}{2\sigma_T^2}\right)}\right)}_{\text{Temperature effects}}, \quad \beta_0,K_h,\hat{T},\sigma^2_T>0.
\end{equation}
The rate of new infections is modeled as 
\begin{equation}\label{FOI}
\Gamma=\beta(h, T)\theta SH + \beta(h,T)\psi A H/N,    
\end{equation}
where $\theta, \psi > 0$ and $N=N(t)$ is the total leaf biomass. The first term employs mass-action incidence to represent long-distance, density-independent transmission by wind-dispersed ascospores, while the second term employs standard incidence to capture localized, frequency-dependent spread by rain-splash-dispersed conidia. 

The environmental dependence incorporates saturation at high humidity and a Gaussian peak around the optimal temperature $\hat{T}$, see Figure \ref{guass}, providing a biologically realistic representation of transmission under varying climatic conditions \cite{bari,ramirez}.

\subsection{Mate limitation in ascospores}
\textit{Pseudocercospora fijiensis} is heterothallic, requiring two compatible mating types (MAT1-1 and MAT1-2, typically in balanced ratios in established populations) for successful sexual reproduction. Ascospores form in pseudothecia on necrotic tissue following fertilization of receptive hyphae by spermatia of the opposite mating type. This process requires physical proximity between compatible mycelia within or between lesions.

At high densities of infected leaf biomass, mating opportunities are abundant, allowing ascospore production to approach its maximum rate. At low densities, such as during early invasion, post-sanitation recovery, or in small or fragmented plantations, the encounter rate between compatible mating types declines sharply, introducing an Allee effect due to mate limitation. The per-unit-biomass rate of successful mating, $g(I)$, therefore decreases nonlinearly with declining density.

The rate of new ascospore production is thus $\alpha g(I(t)) I(t)$, where $\alpha > 0$ and
\begin{equation}\label{eq: function g}
g(I(t))= \frac{\lambda I(t)}{1+\lambda I(t)},\quad\lambda>0.
\end{equation}
When $  I(t) \gg 1/\lambda  $ (high density), $  g(I(t)) \approx 1  $, yielding near-maximal production. When $  I(t) \ll 1/\lambda  $ (low density), $  g(I(t)) \approx \lambda I(t)  $, so ascospore production scales quadratically as $  \alpha \lambda I(t)^2  $, reflecting the pairwise requirement for compatible mates.

Conidia are produced at a constant rate $\eta$ per unit infected biomass. Infections initially produce latent lesions unable to generate spores until maturity (rate $\gamma$). Cultural practices, e.g., leaf removal and sanitation, reduce infected leaf biomass and associated spores at a constant rate $\rho$.

\begin{figure}
    \centering
    \includegraphics[scale = 0.25]{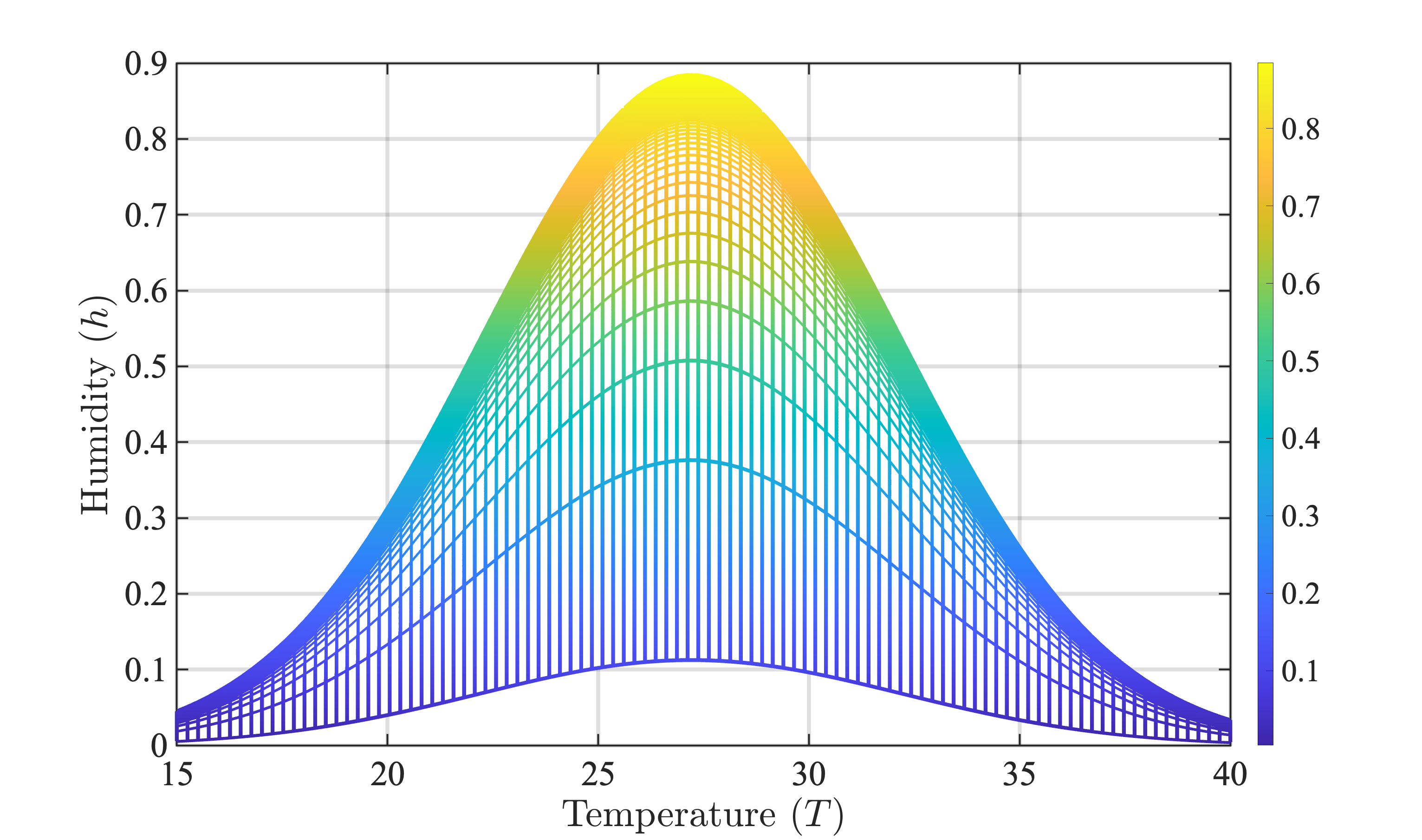}
    \caption{Effect of temperature and relative humidity on the transmission rate $\beta(h,T)$. The plot shows how the transmission rate varies with temperature for various humidity levels. The Gaussian profile with respect to temperature reflects an optimal temperature ($\hat{T} = 27.2^\circ C$) for pathogen activity, beyond which $\beta(h, T)$ declines. The logistic increase with humidity saturates as humidity approaches 1. This captures the biological necessity of warm and moist environments for fungal spore germination and successful infection.}
    \label{guass}
\end{figure}

\subsection{The model}
Figure \ref{flowchart} presents a schematic of the compartmental dynamics. Combining the above assumptions yields the following system of differential equations
\begin{align}\label{eqn1}
    \text{\textbf{Host population}}&
    \begin{cases}
        \dfrac{dH}{dt}=\kappa\Lambda - \beta(h,T)(\theta S+\psi\frac{A}{N})H - \mu H,\\[10pt]
	\dfrac{dR}{dt}=(1-\kappa)\Lambda - (1-\delta)\beta(h,T)(\theta S+\psi\frac{A}{N})R - \mu R,\\[10pt]
	\dfrac{dE}{dt}=\beta(h,T)(\theta S+\psi \frac{A}{N})(H+(1-\delta)R) - (\gamma+\mu)E,\\[10pt]
	\dfrac{dI}{dt}=\gamma E - (\mu+\rho)I.
    \end{cases}\\[10pt]
    \text{\textbf{Pathogen population}}&
    \begin{cases}
        \dfrac{dA}{dt}=\eta I - (\mu_P+\rho)A,\\[10pt]
	\dfrac{dS}{dt}=\alpha g(I)I - (\mu_P+\rho)S.
    \end{cases}
\end{align}
The model accounts for resistance efficacy to the disease, $\delta$, $0\leq \delta\leq 1$, which quantifies the ability of resistant leaves to suppress spore penetration and germination.
These equations capture the essential biological and environmental drivers of BSD transmission while incorporating resistance, dual spore pathways, and density-dependent sexual reproduction. Parameter definitions, dimensions, and values are provided in Tables 1 and 2.

\begin{figure}
    \centering
    \includegraphics[scale =0.45]{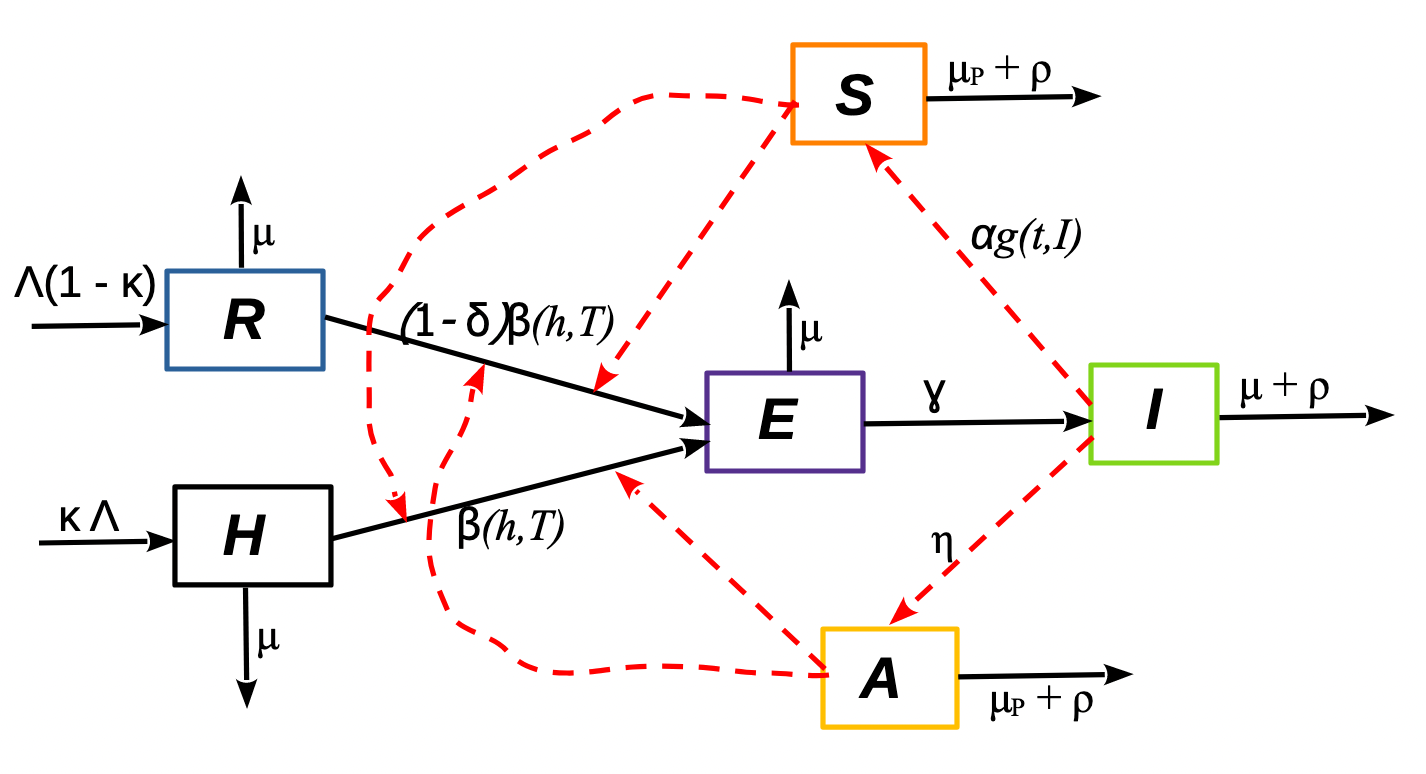}
    \caption{Schematic diagram showing the dynamics of the proposed Black Sigatoka Disease model. Solid arrows represent the progression of biomass of leaves through the disease; dashed arrows indicate the interaction between leaves and spores.}
    \label{flowchart}
\end{figure}

\begin{table}[!ht]
    \caption{The BSD model's parameters and their respective descriptions and dimensions.}
    \label{table1}
    \begin{tabular}{lllll@{}}
        \toprule
        &Description & Dimension\\
        \midrule
        Variables & \\
        $H$ & Biomass of susceptible leaves & Leaves per ha\\
        $R$ & Biomass of resistant leaves   & Leaves per ha\\
        $E$ & Biomass of exposed leaves     & Leaves per ha\\
        $I$ & Biomass of infected leaves    & Leaves per ha\\
        $A$ & Biomass of asexual spores     & Spores per unit of leaf biomass\\
        $S$ & Biomass of sexual spores      & Spores per unit of leaf biomass\\
        Parameters & \\
        $\Lambda$  & Recruitment or replanting rate of healthy leaves & Leaves per day\\
        $\kappa$  & Fraction of susceptible recruited plants & Dimensionless \\
        $\delta$   & Efficacy rate of leaf resistance & Dimensionless \\
        $\beta_0$ & Base transmission rate & day$^{-1}$\\
        $\beta(h,T)$  & Effective disease transmission rate & day$^{-1}$\\
        $\gamma$   & Latency period & day$^{-1}$\\
        $\theta$    & Sexual spore dispersal rate & (spores$\times$day)$^{-1}$\\
        $\psi$ & Asexual spore dispersal rate & (spores$\times$day)$^{-1}$\\
        $\eta$   & Asexual spore production rate & Spores per leaf per day\\
        $\alpha$ & Sexual spore production rate & Spores per leaf per day\\
        $\mu$      & Natural decay of leaves & day$^{-1}$\\
        $\mu_P$    & Natural mortality of pathogen & day$^{-1}$\\
        $\rho$     & Mortality rate of infected leaves, sexual & day$^{-1}$\\
        & and asexual spore due to cultural practices&\\
        $\lambda$  & Maximum achievable mating rate & Mating per lesion\\
        $\sigma_T$ & Temperature tolerance & degrees Celsius ($^{\circ}$C)\\
        $h$        & Humidity level & Dimensionless\\
        $K_h$      & Half-saturation constant for humidity & Dimensionless\\
        $T$  & Temperature & degrees Celsius ($^{\circ}$C)\\
        $\hat{T}$  & Optimal temperature & degrees Celsius ($^{\circ}$C)\\
        \bottomrule
    \end{tabular}
\end{table}

\begin{table}[ht]
	\centering
	\caption{Parameter values of the BSD model \eqref{eqn1} with description in Table \ref{table1}.}
	\begin{tabular}{ccc||ccc}
		\hline
		Parameters&Values&References & Parameters&Values&References\\
		\hline
		  \midrule
            $\Lambda$ & 20 & \cite{agouanet} & $\mu$ & 0.01 & \cite{ravigne} \\
            $\kappa$ & 0.5 & \cite{agouanet} & $\mu_P$ & 0.1 & Assumed\\
            $\delta$ & 0.9 & \cite{agouanet} & $\rho$ & 0.5 & Assumed \\
            $\beta_0$ & 0.9 & Assumed & $\lambda$ & 4 & \cite{agouanet}\\
            $\gamma$ & 0.01 & \cite{ravigne} & $\sigma_T$ & 5    & Assumed\\
            $\eta$ & 20   & \cite{ravigne} & $h$        & 0.9  & Assumed\\
            $\alpha$ & 50  & Assumed & $\hat{T}$  & 27.2 & \cite{bebber}\\
            $\theta$ & 48   & \cite{agouanet} & $T$ & 30.3 & \cite{bebber}\\
            $\psi$ & 20   & \cite{agouanet} & $K_h$      & 0.7  & Assumed \\
		\hline
	\end{tabular}
	\label{table2}
\end{table}
Parameters marked ``Assumed" in Table \ref{table2} were selected to remain consistent with biologically plausible ranges reported in related fungal plant-pathogen modeling studies, particularly those addressing transmission, environmental forcing, and control processes in foliar diseases; the robustness of the model outcomes to these assumptions is subsequently assessed through local, global, and variance-based sensitivity analysis.

\clearpage
\section{Basic Reproduction Number}\label{sec:R0}
\subsection{BSD-free and BSD-persistent equilibria}\label{subsec:DFE}
The disease-free equilibrium corresponds to the complete absence of infection. At this state, the biomass of exposed, infected, and spore compartments must vanish, so $E = I = A = S = 0$. Setting the time derivatives of the host compartments to zero in the model equations (2.4) yields the steady-state values for susceptible and resistant leaf biomass:
$$H^0 = \frac{\kappa \Lambda}{\mu}, \quad R^0 = \frac{(1 - \kappa) \Lambda}{\mu}.$$
Thus, the disease-free equilibrium is
\begin{equation}\label{eqn2}
\mathcal{E}^0 = \left( H^0, R^0, E^0, I^0, A^0, S^0 \right) = \left( \frac{\kappa \Lambda}{\mu}, \frac{(1 - \kappa) \Lambda}{\mu}, 0, 0, 0, 0 \right).
\end{equation}
The total leaf biomass density at the disease-free state is therefore
$$N^0 = H^0 + R^0 = \frac{\Lambda}{\mu}.$$

The disease-persistent (endemic) equilibrium, denoted $\mathcal{E}^* = (H^*, R^*, E^*, I^*, A^*, S^*)$, represents a long-term steady state in which the infection remains persistent within the population and infected and susceptible leaf biomass coexist indefinitely. Identifying the conditions under which this equilibrium exists and is stable enables the prediction of outbreak potential and the evaluation of control interventions in banana plantations. 

The BSD-persistent equilibrium is determined by setting the time derivatives in the BSD model (2.4) and (2.5) equal to zero. This balances the rate of new infections against the removal of infected biomass and the decay of pathogens. Solving the resulting algebraic system gives
\begin{equation}\label{eqn2a}
    \mathcal{E}^*=\left\{\frac{\kappa\Lambda}{\beta(h,T)\Theta_1^*+\mu},\frac{(1-\kappa)\Lambda}{(1-\delta)\beta(h,T)\Theta_1^*+\mu},\frac{\beta(h,T)\Theta_1^*\Theta_2^*}{\Delta_1},\frac{\gamma E^*}{\Delta_2},\frac{\eta\gamma E^*}{\Delta_2\Delta_3},\frac{\alpha\gamma^2\lambda (E^*)^2}{\Delta_2\Delta_3(\Delta_2+\gamma\lambda E^*)}\right\},
\end{equation}
where 
\[\Theta_1^*=\theta S^*+\psi\frac{A^*}{N^*}, \,\Theta_2^*=H^*+(1-\delta)R^*,\, \Delta_1=\gamma+\mu,\,\Delta_2=\mu+\rho, \, \Delta_3=\mu_P+\rho.\]
\subsection{\texorpdfstring{$\mathcal{R}_0$}{} for BSD}
The basic reproduction number, denoted $\mathcal{R}_0$, is a threshold parameter in epidemiology that quantifies the expected number of secondary infections produced by a single infected individual introduced into a fully susceptible population. It determines whether an infection can invade and establish itself ($\mathcal{R}_0 > 1$) or is expected to die out ($\mathcal{R}_0 < 1$). Computing $\mathcal{R}_0$ is essential for assessing invasion risk, designing control strategies, and identifying interventions for reducing disease transmission.

A standard and rigorous approach for deriving $\mathcal{R}_0$ is the next-generation matrix method \cite{afful, diekmann, van}. This technique constructs two matrices: the transmission Jacobian matrix $\mathcal{F}$, which captures the rate of new infections in each infected compartment, and the transition Jacobian matrix $\mathcal{V}$, which describes the rates of progression or removal among infected compartments. The basic reproduction number is then the spectral radius (dominant eigenvalue) of the next-generation matrix $\mathcal{FV}^{-1}$.

In our model, the infected compartments are the exposed $E$, the infected $I$, and the asexual spores (conidia), $A$. Sexual spores (ascospores, $S$) are excluded from the computation of $\mathcal{R}_0$. Near the disease-free equilibrium, the infected leaf biomass $I \approx 0$, so the mating function satisfies $g(I) I \sim I^2$. Consequently, sexual spore production is a higher-order term that vanishes quadratically and does not contribute to the leading-order invasion dynamics in the linearized system.
We therefore consider the vector of infected states $\mathbf{X} = (E, I, A)^\top$ and express the dynamics of these compartments as
$$\frac{d\mathbf{X}}{dt} = F(\mathbf{X}) - V(\mathbf{X}),$$
where the new infection matrix is
$$F(\mathbf{X}) = \begin{bmatrix}
\beta(h,T) \psi \frac{A}{N} \left[ H + (1 - \delta) R \right] \\
0 \\
0
\end{bmatrix},$$
and the transition matrix is
$$V(\mathbf{X}) = \begin{bmatrix}
(\gamma + \mu) E \\
- \gamma E + (\mu + \rho) I \\
- \eta I + (\mu_P + \rho) A
\end{bmatrix}.$$
To compute $\mathcal{R}_0$, evaluate the Jacobians of $F$ and $V$, $\mathcal{F}$ and $\mathcal{V}$ respectively, at the disease-free equilibrium $\mathcal{E}^0 = (H^0, R^0, 0, 0, 0, 0)$, where $H^0 = \kappa \Lambda / \mu $ and $R^0 = (1 - \kappa) \Lambda / \mu$, 
\begin{align}\label{eqn4}\nonumber
    \mathcal{F}(\mathcal{E}^0)
    &= \begin{bmatrix}
        0&0 & \beta\psi\left[\kappa\delta+(1-\delta)\right]\\
        0 & 0 & 0\\
        0 & 0 & 0
    \end{bmatrix},
    \qquad\mathcal{V}(\mathcal{E}^0)
    = \begin{bmatrix}
        \gamma+\mu & 0 & 0 \\
        -\gamma  & \mu+\rho & 0 \\
        0 & -\eta  & \mu_P+\rho
    \end{bmatrix}.
\end{align}
The basic reproduction number is then
$$\mathcal{R}_0 = \sigma\bigl(\mathcal{F}\mathcal{V}^{-1} \bigr),$$
where $  \sigma(\cdot)  $ denotes the spectral radius. In our case,
\[\mathcal{V}^{-1}(\mathcal{E}^0)
    = \begin{bmatrix}
        \dfrac{1}{\gamma+\mu} & 0 & 0 \\
        \dfrac{\gamma}{(\gamma+\mu)(\mu+\rho)}&\dfrac{1}{\mu+\rho}&0 \\
        \dfrac{\eta\gamma}{(\gamma+\mu)(\mu+\rho)(\mu_P+\rho)} & \dfrac{\eta}{(\mu+\rho)(\mu_P+\rho)} & \dfrac{1}{(\mu_P+\rho)}
    \end{bmatrix}.\]
It is straightforward to see that the matrix product $\mathcal{F V}^{-1}$ has a single nonzero eigenvalue,
\begin{equation}\label{eqn5}
    \mathcal{R}_0=\frac{\beta(h,T)\psi\eta\gamma[\delta\kappa+(1-\delta)]}{(\gamma+\mu)(\mu_P+\rho)(\mu+\rho)}.
\end{equation}
This expression for $  \mathcal{R}_0  $ reflects the expected number of new infections (primarily via conidia transmission in the early invasion phase) produced by one initially infected unit in a fully susceptible leaf population, accounting for latency, asexual spore production, and decay processes. In regimes where mate limitation is negligible at invasion, i.e., sexual reproduction contributes negligibly near $I = 0$, $\mathcal{R}_0$ is dominated by the asexual transmission pathway. However, the full nonlinear model, including sexual spore production and the associated Allee effect, gives rise to a backward bifurcation, as demonstrated subsequently.

\subsection{Disease prevalence for \texorpdfstring{$\mathcal{R}_0<1$}{R0<1}}
As mentioned in the previous subsection, the infections produced by ascospores are not captured in the computation of $\mathcal{R}_0$ because, for small infected leaf biomass, the corresponding term of the force of infection is nonlinear, $\alpha g(I)I\sim \alpha\lambda I^2$. Here, we show that the disease can successfully invade a healthy population despite having $\mathcal{R}_0<1$.

Let us denote by $\Phi$ the effective host biomass available for infection at the BSD-free equilibrium,
\begin{equation}
    \Phi=\frac{H^0+(1-\delta)R^0}{N^0}=\kappa+(1-\delta)(1-\kappa)=\delta \kappa+(1-\delta).
\end{equation}
During an early invasion phase, disease progression is slow, allowing us to assume that the exposed leaves, conidia, and ascospores are at a quasi-steady state, that is, $dE/dt=dA/dt=dS/dt=0$. Therefore,
\begin{equation}
    E=\frac{\Gamma\times\left(H^0+(1-\delta)R^0\right)}{\gamma+\mu},\quad A=\frac{\eta I}{\mu_P+\rho},\quad S=\frac{\alpha g(I)I}{\mu_P+\rho}.
\end{equation}
Substituting the last two approximations into the force of infection \eqref{FOI} yields
\begin{equation}
    \Gamma=\beta(h,T)\left(\theta\alpha g(I)+\frac{\psi\eta}{N^0}\right)\frac{I}{\mu_P+\rho}.
\end{equation}
Substituting the approximation to $E$ in the equation for $dI/dt$ and recalling that $N^0=\Lambda/\mu$ gives
\begin{align*}
    \frac{1}{I}\frac{dI}{dt} &= \gamma \frac{E}{I} -(\mu+\rho)\\
    &= \gamma \beta(h,T)\left(\theta\alpha g(I)+\frac{\psi\eta}{N^0}\right)\frac{1}{\mu_P+\rho}\frac{\left(H^0+(1-\delta)R^0\right)}{\gamma+\mu}-(\mu+\rho)\\[10pt]
    &=(\gamma \beta(h,T) \theta \alpha g(I) N^0+\gamma \beta(h,T)\psi\eta)\frac{\Phi}{(\mu_P+\rho)(\gamma+\mu)}-(\mu+\rho) \\[10pt]
    &=\underbrace{\frac{\Delta\theta\alpha g(I)}{\mu_P+\rho}}_{NL}+\underbrace{(\mathcal{R}_0-1)(\mu+\rho)}_{L},
\end{align*}
where  
\begin{equation*}
    \Delta=\frac{\beta(h,T)\Lambda\gamma\Phi}{\mu(\gamma+\mu)}
\end{equation*}
and the nonlinear and linear terms of the equation are denoted by $NL$ and $L$, respectively.
If the nonlinear term contribution is neglected, because $I$ is very small, then all the growth of infected biomass comes from the linear contribution when $L>0$, or equivalently, when $\mathcal{R}_0>1$, as expected.  

Let us consider now the effects of the nonlinear term $NL$ and answer the question if it is possible to have a disease invasion even if $\mathcal{R}_0<1$.
That would be the case if the inequality
\begin{equation}
    g(I)>\underbrace{\frac{(1-\mathcal{R}_0)(\mu+\rho)(\mu_P+\rho)}{\Delta\theta\alpha }}_{D}
\end{equation}
is satisfied. By recalling the definition of $g(I)$ given in \eqref{eq: function g}, we obtain the following equivalent condition for a successful disease invasion
\begin{equation}
    I \,\lambda \left(\frac{1}{D}-1\right)>1,
\end{equation}
from which it is necessary that $D$ be less than one. Let $I_c$ be the invasion threshold, i.e., the value of $I$ at which equality is held in the previous relationship. Then
\begin{equation}
    I_c=\frac{D}{\lambda(1-D)}.
\end{equation}
In conclusion, when $\mathcal{R}_0<1$ but $I>I_c$ and $D<1$, there will be disease invasion. 
The biological interpretation of these conditions is that the nonlinear effects of sexual production of ascospores may trigger an outbreak to advance when the threshold of infected biomass $I_c$ is crossed, and ascospore production is sufficiently rapid. This suggests bistability in the model via a backward bifurcation, which is a subcritical transcritical bifurcation.

\subsection{The Backward Bifurcation}
This resilience of the BSD can be better understood by considering a backward bifurcation, in which the disease can maintain stable endemic states even when the basic reproduction number $\mathcal{R}_0$ falls below 1, defying conventional thresholds for eradication. 
In our model, this arises from dual spore transmission pathways and mate limitation in ascospore production, leading to bistability and hysteresis that complicate control strategies. The rigorous analysis draws on Theorem 4.1 from \cite{castillo}, originally applied to tuberculosis dynamics, to confirm subcritical behavior and underscore the need for innovative approaches, such as resistant varieties.

\subsubsection{Numerical results}
Here, we present numerical computations of the equilibria for the model given in Equation \eqref{eqn1}, using the parameter values in Table 2. The results, illustrated in Figure \ref{fig: backwardbif}, clearly show the occurrence of a backward bifurcation. 
In the context of our BSD model, the basic reproductive number, $\mathcal{R}_0$, is an inadequate predictor of disease eradication. While a traditional approach would suggest that reducing $\mathcal{R}_0$ below one is enough to attain eradication, the nature of the \textit{Mycosphaerella fijiensis}, its polycyclic life cycle, the dual spore pathways, capacity of reinfection, and long-distance dispersal may allow the pathogen to persist on a stable branch even when $\mathcal{R}_0<1$. 
The bistable behavior arises predominantly from the nonlinear dynamics in the sexual reproduction cycle, which creates an unstable equilibrium branch in the bifurcation diagram (red dashes).

The bifurcation diagram in Figure \ref{fig: backwardbif} shows two remarkable characteristics. First, the critical value of $\mathcal{R}_0$, for which the fold in the bifurcation curve occurs, is exceptionally low. This indicates that outbreaks of Black Sigatoka disease can emerge under a wide range of epidemiological conditions, even when the basic reproduction number is only slightly above zero. Second, the unstable equilibrium branch remains extremely close to zero for very small positive values of $\mathcal{R}_0$. Consequently, the introduction of even a tiny initial density of infected leaf biomass may be sufficient to trigger a full epidemic outbreak, as the system would lie within the basin of attraction of the stable endemic equilibrium.

These observations are consistent with the fact that BSD has not been successfully eradicated in practice, but on a single occasion, in Australia, between the years 2001 and 2005, after an outbreak in the Tully region of Queensland \cite{peterson, cook}. The program required an aggressive approach that included the complete destruction of infected plantations, the removal and burial of plant materials, the burning of infected leaves, strict quarantine zones, surveillance, and movement controls coordinated by Biosecurity Queensland \cite{henderson, sosnowski}. The Tully Banana Black Sigatoka Eradication program involved the destruction of approximately $95\%$ of banana plantations in the infected zone and a six-month ban on replanting to ensure the pathogen had no living host to maintain its population \cite{henderson}. 
This is consistent with our model results, which suggest that successful management depends on intensive biomass reduction to shift the system into the narrow basin of the stable state, representing a BSD-free equilibrium.

\begin{figure}
    \centering
    \includegraphics[width=0.8\linewidth]{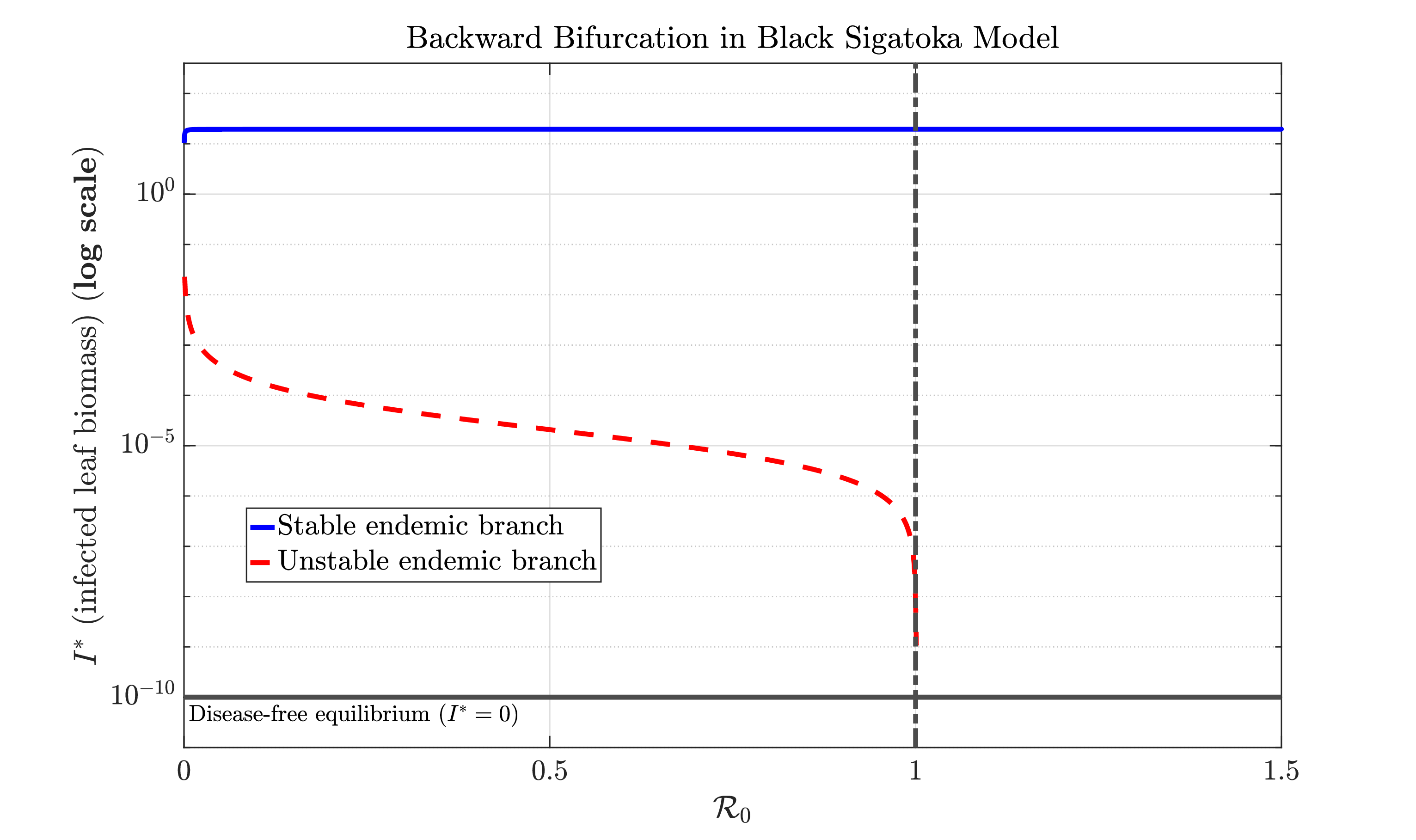}
    \caption{Backward bifurcation in the model \eqref{eqn1}. The stable endemic branch (blue) extends beyond $\mathcal{R}_0=1$. The unstable branch (red dashes) emerges at $\mathcal{R}_0=1$ and bends ``backward". The coexistence of the positive and disease-free equilibria for $\mathcal{R}_0<1$ persists for extremely small values of $\mathcal{R}_0$. The fold on the bifurcation curve for small values of $\mathcal{R}_0$ is not displayed due to computational limitations.}
    \label{fig: backwardbif}
\end{figure}

\subsubsection{Analytical results}
Here, we analytically investigate the presence of a backward bifurcation in the BSD model \eqref{eqn1}. We follow \cite{castillo}, Theorem 4.1, which provides conditions, derived from eigenvectors and second-order partial derivatives, under which a dynamical system undergoes a backward bifurcation at the disease-free equilibrium, when the bifurcation parameter $\beta_0$ passes through $\beta_0^*$, where $\beta_0^*$ is the value that corresponds to $\mathcal{R}_0=1$. 
\begin{thm}
    Let us denote $\bfx=(H,R,E,I,A,S)$ and $\bff=d\bfx/dt$. The bifurcation coefficients $\bfa$ and $\bfb$ are defined as
    \begin{equation*}
        \bfa=\sum_{k,i,j=1}^{6}v_kw_iw_j\frac{\partial^2f_k(\mathcal{E}^0,\beta_0^*)}{\partial x_i\partial x_j},\qquad
        \bfb=\sum_{k,i=1}^{6}v_kw_i\frac{\partial^2f_k(\mathcal{E}^0,\beta_0^*)}{\partial x_i\partial\beta_0^*}.
    \end{equation*}
    Then, the BSD model \eqref{eqn1} undergoes a backward bifurcation at $\beta_0=\beta_0^*$, i.e. $\mathcal{R}_0=1$, when $\bfa>0$, $\bfb>0$ and $\lambda>\lambda^{\text{crit}}$, where
    \begin{equation}
        \lambda^{\text{crit}}=\frac{\beta_0^*\Pi\mu\psi^2\eta^2[(1-\delta)+\delta\kappa(2+\delta)]}{\alpha\theta\Lambda^2(\mu_P+\rho)[\delta\kappa+(1-\delta)]}.
    \end{equation}
\end{thm}
\begin{proof}
    Setting $\mathcal{R}_0=1$ and solving for $\beta_0^*$ in \eqref{eqn5} gives;
    \begin{equation}\label{eqn7}
        \beta_0^*=\frac{(\gamma+\mu)(\mu_P+\rho)(\mu+\rho)}{\Pi\psi\eta\gamma[\delta\kappa+(1-\delta)]},
    \end{equation}
    where $\Pi=\frac{h}{h+K_h}\exp{\left(-\frac{(T-\hat{T})^2}{2\sigma_T^2}\right)}$. The Jacobian matrix of the BSD model \eqref{eqn1} evaluated at the BSD-free equilibrium, $\mathcal{E}^0$ with $\beta_0=\beta_0^*$ is
    \begin{equation}\label{eqn8}
        \mathcal{J}(\mathcal{E}^0,\beta_0^*)=
        \begin{bmatrix}
            -\mu&0&0&0&-\beta_0^*\Pi\psi\kappa&-\frac{\beta_0^*\Pi\theta\kappa\Lambda}{\mu}\\
            0&-\mu&0&0&-\beta_0^*\Pi\psi(1-\delta)(1-\kappa)&-\frac{\beta_0^*\Pi\theta\Lambda(1-\delta)(1-\kappa)}{\mu}\\
            0&0&-(\gamma+\mu)&0&\beta_0^*\Pi\psi[\delta\kappa+(1-\delta)]&\frac{\beta_0^*\Pi\theta\Lambda[\delta\kappa+(1-\delta)]}{\mu}\\
            0&0&\gamma&-(\mu+\rho)&0&0\\
            0&0&0&\eta&-(\mu_P+\rho)&0\\
            0&0&0&0&0&-(\mu_P+\rho)
        \end{bmatrix}.
    \end{equation}
    This matrix has $\lambda_1=\lambda_2=-\mu$, and $\lambda_6=-(\mu_P+\rho)$ as eigenvalues. The remaining eigenvalues are obtained from the sub-matrix
    \begin{equation}\label{eqn9}
        \hat{\mathcal{J}}(\mathcal{E}^0,\beta_0^*)=
        \begin{bmatrix}
            -(\gamma+\mu)&0&\beta_0^*\Pi\psi[\delta\kappa+(1-\delta)]\\
            \gamma&-(\mu+\rho)&0\\
            0&\eta&-(\mu_P+\rho)
        \end{bmatrix},
    \end{equation}
    for which the associated characteristic polynomial at $\chi$ is
    \begin{equation}
        |\chi\bfI-\hat{\mathcal{J}}(\mathcal{E}^0,\beta_0^*)|=(\chi+\gamma+\mu)(\chi+\mu+\rho)(\chi+\mu_P+\rho)-\beta_0^*\Pi\psi\gamma\eta[\delta\kappa+(1-\delta)].
    \end{equation}
    At $\chi=0$, there is an eigenvalue when
    \begin{align*}
        (\gamma+\mu)(\mu+\rho)(\mu_P+\rho)-\beta_0^*\Pi\psi\gamma\eta[\delta\kappa+(1-\delta)]&=0,\\
        \frac{\beta_0^*\Pi\psi\gamma\eta[\delta\kappa+(1-\delta)]}{(\gamma+\mu)(\mu+\rho)(\mu_P+\rho)}&=1,\\
        \mathcal{R}_0&=1.
    \end{align*}
    Therefore, at $\beta_0=\beta_0^*$ that is $\mathcal{R}_0=1$, the Jacobian matrix associated with the BSD model \eqref{eqn1}  evaluated at $\mathcal{E}^0$ has a simple zero eigenvalue (with all other eigenvalues having negative real part). Thus, the Jacobian matrix has a right eigenvector given by $\bfw=(w_1,w_2,w_3,w_4,w_5,w_6)$ and a left eigenvector given by $\bfv=(v_1,v_2,v_3,v_4,v_5,v_6)$ corresponding to the zero eigenvalue. A right eigenvector satisfies
    \begin{equation}\label{eqn10}
        \begin{bmatrix}
            -\mu&0&0&0&-\beta_0^*\Pi\psi\kappa&-\frac{\beta_0^*\Pi\theta\kappa\Lambda}{\mu}\\
            0&-\mu&0&0&-\beta_0^*\Pi\psi(1-\delta)(1-\kappa)&-\frac{\beta_0^*\Pi\theta\Lambda(1-\delta)(1-\kappa)}{\mu}\\
            0&0&-(\gamma+\mu)&0&\beta_0^*\Pi\psi[\delta\kappa+(1-\delta)]&\frac{\beta_0^*\Pi\theta\Lambda[\delta\kappa+(1-\delta)]}{\mu}\\
            0&0&\gamma&-(\mu+\rho)&0&0\\
            0&0&0&\eta&-(\mu_P+\rho)&0\\
            0&0&0&0&0&-(\mu_P+\rho)
        \end{bmatrix}
        \begin{bmatrix}
            w_1\\w_2\\w_3\\w_4\\w_5\\w_6
        \end{bmatrix}
        =
        \begin{bmatrix}
            0\\0\\0\\0\\0\\0
        \end{bmatrix}
    \end{equation}
    Solving for $w_1$, $w_2$, $w_3$, and $w_5$ in terms of $w_4$ yields
    \begin{align*}
        w_1&=-\frac{\beta_0^*\Pi\psi\kappa\eta}{\mu(\mu_P+\rho)}w_4,\;\;
        w_2=-\frac{\beta_0^*\Pi\psi\eta(1-\delta)(1-\kappa)}{\mu(\mu_P+\rho)}w_4,\;\;
        w_3=\frac{\mu+\rho}{\gamma}w_4,\\
        w_4&>0,\;\;
        w_5=\frac{\eta}{\mu_P+\rho}w_4,\;\;
        w_6=0.
    \end{align*}
    A left eigenvector satisfies
    \begin{equation}\label{eqn11}
        \begin{bmatrix}
            -\mu&0&0&0&0&0\\
            0&-\mu&0&0&0&0\\
            0&0&-(\gamma+\mu)&\gamma&0&0\\
            0&0&0&-(\mu+\rho)&\eta&0\\
            -\beta_0^*\Pi\psi\kappa&-\beta_0^*\Pi\psi(1-\delta)(1-\kappa)&\beta_0^*\Pi\psi[\delta\kappa+(1-\delta)]&0&-(\mu_P+\rho)&0\\
            -\frac{\beta_0^*\Pi\theta\kappa\Lambda}{\mu}&-\frac{\beta_0^*\Pi\theta\Lambda(1-\delta)(1-\kappa)}{\mu}&\frac{\beta_0^*\Pi\theta\Lambda[\delta\kappa+(1-\delta)]}{\mu}&0&0&-(\mu_P+\rho)
        \end{bmatrix}
        \begin{bmatrix}
            v_1\\v_2\\v_3\\v_4\\v_5\\v_6
        \end{bmatrix}
        =
        \begin{bmatrix}
            0\\0\\0\\0\\0\\0
        \end{bmatrix}
    \end{equation}
    Solving for $v_4$, $v_5$, and $v_6$ in terms of $v_3$ gives
    \begin{align*}
        v_1&=0,\;
        v_2=0,\;
        v_3>0,\;\;
        v_4=\frac{\gamma+\mu}{\gamma}v_3,\;\;
        v_5=\frac{(\mu+\rho)(\gamma+\mu)}{\eta\gamma}v_3,\;\;
        v_6=\frac{\beta_0^*\Pi\theta\Lambda[\delta\kappa+(1-\delta)]}{\mu(\mu_P+\rho)}v_3.
    \end{align*}
    Since $v_1$ and $v_2$ are zero, we only compute the second-order partial derivatives of $f_3$, $f_4$, $f_5$,  and $f_6$,
    \begin{align*}
        \frac{\partial^2f_3(\mathcal{E}^0,\beta_0^*)}{\partial H\partial S}&=\frac{\partial^2f_3(\mathcal{E}^0,\beta_0^*)}{\partial S\partial H}=\beta_0^*\Pi\theta,\\
        \frac{\partial^2f_3(\mathcal{E}^0,\beta_0^*)}{\partial H\partial A}&=\frac{\partial^2f_3(\mathcal{E}^0,\beta_0^*)}{\partial A\partial H}=\frac{\beta_0^*\Pi\psi\mu}{\Lambda},\\
        \frac{\partial^2f_3(\mathcal{E}^0,\beta_0^*)}{\partial R\partial S}&=\frac{\partial^2f_3(\mathcal{E}^0,\beta_0^*)}{\partial S\partial R}=\beta_0^*\Pi\theta(1-\delta),\\
        \frac{\partial^2f_3(\mathcal{E}^0,\beta_0^*)}{\partial R\partial A}&=\frac{\partial^2f_3(\mathcal{E}^0,\beta_0^*)}{\partial A\partial R}=\frac{\beta_0^*\Pi\psi\mu(1-\delta)}{\Lambda},\\
        \frac{\partial^2f_6(\mathcal{E}^0,\beta_0^*)}{\partial I^2}&=2\alpha\lambda,\\
        \frac{\partial^2f_3(\mathcal{E}^0,\beta_0^*)}{\partial A\partial\beta_0^*}&=\Pi\psi\mu[\delta\kappa+(1-\delta)],\\
        \frac{\partial^2f_3(\mathcal{E}^0,\beta_0^*)}{\partial S\partial\beta_0^*}&=\frac{\Pi\theta\Lambda[\delta\kappa+(1-\delta)]}{\mu}.
    \end{align*}
    Using these expressions in the definitions of $\bfa$ gives 
    \begin{equation}
        \bfa=2v_3w_5\frac{\beta_0^*\Pi\psi\mu}{\Lambda}\left[w_1+(1-\delta)w_2\right]+2v_6w_4^2\alpha\lambda.
    \end{equation}
    Substituting the values of $w_1$, $w_2$, $w_5$, and setting $w_4=1$ yields
    \begin{equation}
        \bfa=\frac{2v_3\alpha\lambda\beta_0^*\Pi\theta\Lambda[\delta\kappa+(1-\delta)]}{\mu(\mu_P+\rho)}-\frac{2v_3(\beta_0^*\Pi\psi\eta)^2[(1-\delta)+\delta\kappa(2+\delta)]}{\Lambda(\mu_P+\rho)}.
    \end{equation}
    Therefore, if 
    \begin{equation}\label{eqn12}
        \lambda>\underbrace{\frac{\beta_0^*\Pi\mu\psi^2\eta^2[(1-\delta)+\delta\kappa(2+\delta)]}{\alpha\theta\Lambda^2(\mu_P+\rho)[\delta\kappa+(1-\delta)]}}_{\lambda^{\text{crit}}}
    \end{equation}
    then $\bfa>0$. Moreover, 
    \begin{equation}
        \bfb=v_3w_5\Pi\psi[\delta\kappa+(1-\delta)]=\frac{\Pi\psi\eta[\delta\kappa+(1-\delta)]}{\mu_P+\rho}v_3>0.
    \end{equation}
    By applying Theorem 4.1 of Castillo-Chavez and Song, we conclude that the BSD model \eqref{eqn1} undergoes a backward bifurcation. 
    \end{proof}
    
    Condition \eqref{eqn12} means that the disease can persist at low values of $\mathcal{R}_0$ because the high biomass of infected leaves in an established epidemic overcomes the initial difficulty of mating spores. Using the parameter values from Table \ref{table2}, we have $\bfa=1.4700\times10^7$ and $\lambda>\lambda^{\text{crit}}=0.0011$. This indicates that sexual reproduction of the pathogen occurs well above the threshold required to induce backward bifurcation. Since $\lambda>\lambda^{\text{crit}}$, the mate limitation in ascospores is a powerful driver of the epidemic's persistence. Thus, the Black Sigatoka Disease may persist even when $\mathcal{R}_0<1$, and reducing $\mathcal{R}_0$ below unity may not guarantee disease eradication.

\section{Sensitivity analysis}\label{sec:sensitivity}
Sensitivity analysis is an analytical tool used to quantify the relative impacts of uncertainties in the input parameters for endemic-infected leaf biomass ($I^*$) and their subsequent effects on disease spread. Thus, sensitivity analysis can determine how input variability affects output variability. 
\subsection{Local analysis using the Normalized Forward Sensitivity Index}
To quantify the effect of parameter perturbations on $I^*$, we use the normalized forward sensitivity index introduced by Chitnis et al. \cite{chitnis}. For a parameter $\xi$, the index is defined by:
\begin{equation}\label{eqn6}
    \Upsilon_{\xi}^{I^*}=\frac{\partial I^*}{\partial\xi}\times\frac{\xi}{I^*}.
\end{equation}
This gives the relative change in $I^*$ per unit relative change in the parameter $\xi$. A positive index indicates that the $I^*$ increases as the parameter value increases, whereas a negative index indicates that $I^*$ decreases as the parameter value increases. This indicates that parameters with negative indices may mitigate disease spread in the population as their values increase.

Figure \ref{nfsi} shows that the fraction of newly recruited susceptible leaves, $\kappa$, is the most influential parameter with an index of 1.47. Thus, a $1\%$ increase in the fraction of newly recruited susceptible leaves leads to a $1.47\%$ increase in $I^*$. In contrast, parameters such as $\beta_0$, $\Lambda$, $\rho$, $\mu_P$, $\mu$, and other parameters in $I^*$ have indices close to zero, indicating weak local influence near the baseline parameter set. Thus, the recruitment structure and host susceptibility drive the $I^*$ far more strongly than most other biological processes do.

\begin{figure}[!htbp]
	\centering
	\includegraphics[width=0.7\linewidth]{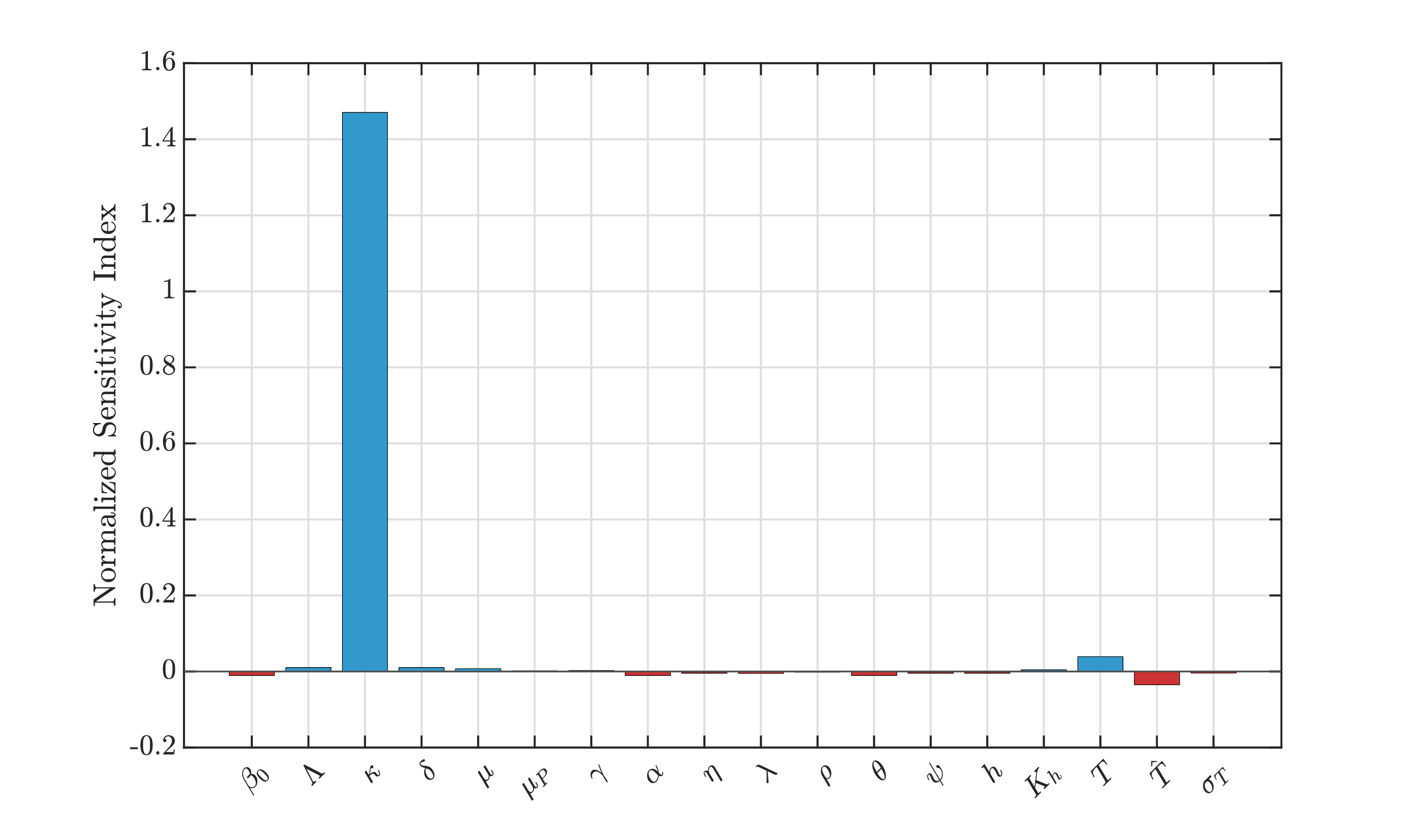}
    \caption{The normalized forward sensitivity indices of the parameters in the endemic-infected leaf biomass ($I^*$) generated using the parameter values stated in Table \ref{table2}.}
    \label{nfsi}
\end{figure}

\subsection{Global analysis using the Partial Rank Correlation Coefficient (PRCC)}
To assess the global influence of model parameters on the endemic-infected leaf biomass $I^*$, we use Partial Rank Correlation Coefficients (PRCCs) combined with Latin Hypercube Sampling (LHS). This approach evaluates the monotonic relationship between each parameter and $I^*$ while accounting for simultaneous variations in all other parameters. PRCC values close to 1 (-1) indicate a strong positive (negative) global influence on $I^*$, whereas values near zero indicate weak or negligible effects.

Figure \ref{prcc}, the host recruitment parameters $\Lambda$ and $\kappa$ have the strongest positive correlations, indicating that $I^*$ is primarily driven by the continuous supply of susceptible leaf biomass. The progression rate $\gamma$ also has a moderate global positive effect on $I^*$. In contrast, the cultural practices parameter $\rho$ shows a negative correlation, indicating that the removal of infected biomass is the most effective control mechanism. The natural leaf decay rate $\mu$ also contributes negatively by reducing available host biomass.

The violin plots in Figure \ref{prcc_vio} illustrate the robustness of these results. Parameters such as $\Lambda$, $\kappa$, $\gamma$, $\rho$, and $\mu$ display narrow distributions centered away from zero, indicating consistent and statistically stable effects across simulations, whereas parameters with distributions concentrated near zero exhibit minimal influence. The results confirm that host recruitment dynamics and cultural practices dominate the global behavior of $I^*$ in the BSD model \eqref{eqn1}.

\begin{figure}[!htbp]
    \begin{subfigure}{0.55\textwidth}
	\centering
	\includegraphics[width=1.0\linewidth]{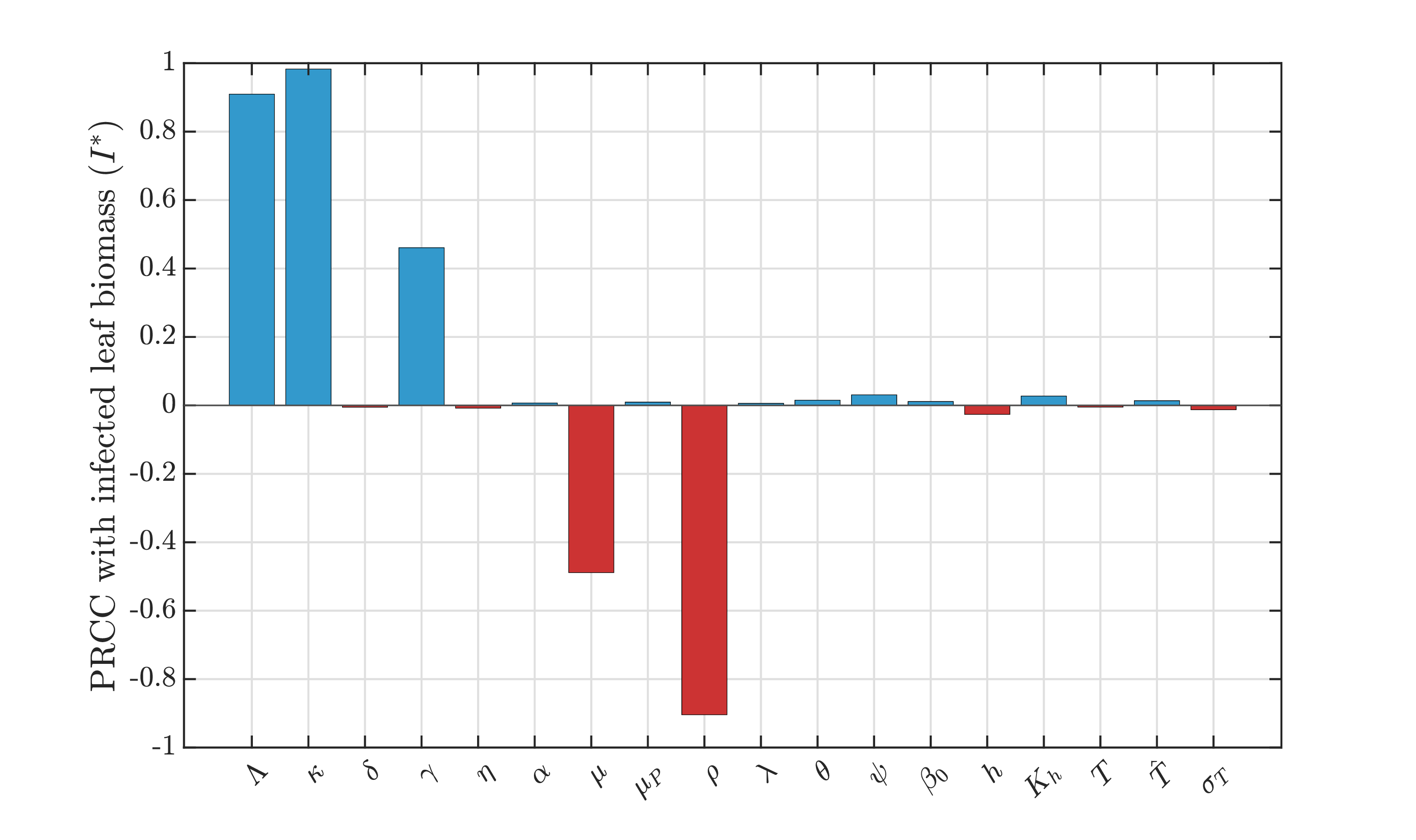}
        \caption{}
        \label{prcc}
    \end{subfigure}
    \begin{subfigure}{0.55\textwidth}
	\centering
	\includegraphics[width=1.0\linewidth]{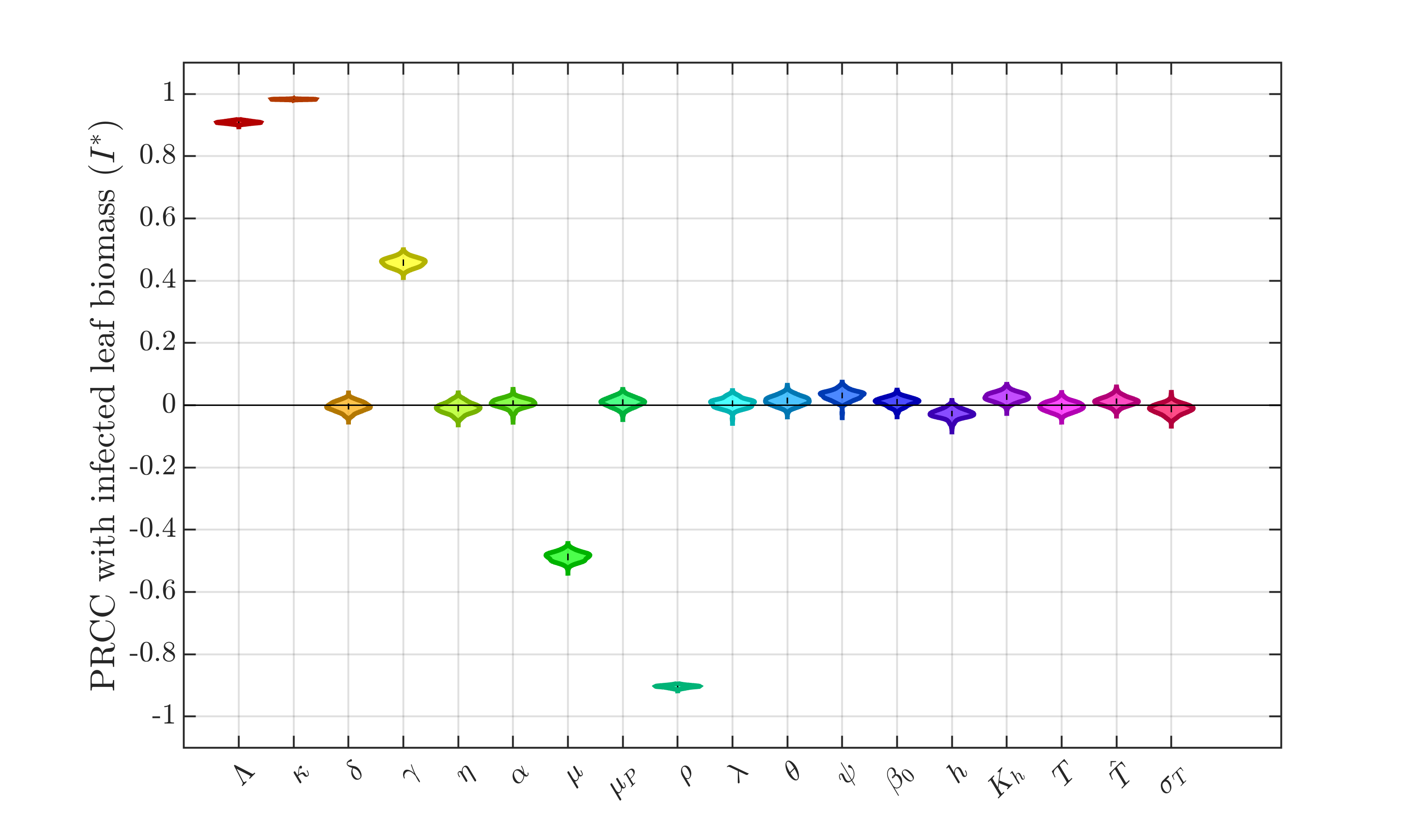}
        \caption{}
        \label{prcc_vio}
    \end{subfigure}
    \caption{\textbf{(a)} The Partial Rank Correlation Coefficients (PRCCs) results capturing the global sensitivity indices of the parameters in the infected biomass equilibrium ($I^*$). \textbf{(b)} The Violin plot of the Partial Rank Correlation Coefficients (PRCCs) values for the parameters in $I^*$ (Generated distribution around PRCC values.) Generated using the parameter values stated in Table \ref{table2}.}
    \label{fig3}
\end{figure}

\clearpage
\section{Stochastic modeling using the Gillespie algorithm}\label{sec:sto}
To incorporate intrinsic randomness in disease transmission, we formulate a stochastic version of the BSD model using the Gillespie Algorithm, also known as the Stochastic Simulation Algorithm (SSA), \cite{gillespie}. In this framework, the biomass variables are treated as discrete ``leaf units" or effective spore units, so that the stochastic process approximates plantation-scale dynamics while retaining the event-by-event structure of the SSA. Although leaf biomass is continuous in reality, this discretization provides a reasonably effective approximation for finite host populations and enables the quantification of demographic stochasticity near eradication thresholds, where deterministic models may fail to capture variability accurately.

The stochastic model tracks the state of the vector process $\bfY=(H, R, E, I, A, S)$, where leaf and spore biomasses are represented as in the previous section. The system's evolution is governed by a set of discrete events, each associated with a propensity function $a_i(\bfY)$ that defines the probability of occurrence in an infinitesimal time interval, and a corresponding stoichiometric vector $\nu_i$. The full list of events and propensities is provided in Table \ref{table4}.

At each step, the Gillespie algorithm computes the total propensity, determines the time to the next event, and selects which event occurs based on probabilistic sampling. The system is then updated accordingly, producing stochastic trajectories that reflect variability in disease dynamics. This framework enables comparison with deterministic predictions and provides insight into fluctuations, persistence, and eradication behavior in the presence of stochastic effects.

\begin{table}[!ht]
    \caption{List of events, propensity functions, and state changes used in the stochastic simulation of Black Sigatoka Disease dynamics using the Gillespie algorithm.}
    \label{table4}
    \begin{tabular}{lllll@{}}
        \toprule
         $i$ & Event & Propensity function $a_i(\bfY)$ & State change $\nu_i$ \\
        \midrule
        1 & Recruitment of susceptible leaves & $\kappa \Lambda$ & $H \rightarrow H + 1$\\
        2 & Recruitment of resistant leaves & $(1 - \kappa)\Lambda$  & $R \rightarrow R + 1$\\
        3 & Natural decay of susceptible leaves & $\mu H$  & $H \rightarrow H - 1$\\
        4 & Natural decay of resistant leaves & $\mu R$ & $R\rightarrow R-1$\\
        5 & Infection of susceptible leaves via asexual spores & $\beta(h,T)\psi\frac{A}{N}H$ & $H\rightarrow H-1$,\\
        &&$E\rightarrow E+1$\\
        6 & Infection of susceptible leaves via sexual spores & $\beta(h,T)\theta SH$ & $H\rightarrow H-1$,\\ 
        &&$E\rightarrow E+1$\\
        7 & Infection of resistant leaves via asexual spores & $(1-\delta)\beta(h,T)\psi\frac{A}{N}R$ & $R\rightarrow R-1$,\\ 
        &&$E\rightarrow E+1$\\
        8 & Infection of resistant leaves via sexual spores & $(1-\delta)\beta(h,T)\theta SR$ & $R\rightarrow R-1$,\\ 
        &&$E\rightarrow E+1$\\
        9 & Progression from exposed to infected & $\gamma E$ & $E\rightarrow E-1$,\\ 
        &&$I\rightarrow I+1$\\
        10 & Natural decay of exposed leaves & $\mu E$ & $E\rightarrow E-1$\\
        11 & Natural decay of infected leaves & $(\mu+\rho)I$ & $I\rightarrow I-1$\\
        12 & Asexual spore production & $\eta I$ & $A\rightarrow A+1$\\
        13 & Sexual spore production & $\alpha\frac{\lambda I}{1+\lambda I}I$ & $S\rightarrow S+1$\\
        14 & Loss of asexual spores & $(\mu_P+\rho)A$ & $A\rightarrow A-1$\\
        15 & Loss of sexual spores & $(\mu_P+\rho)S$ & $S\rightarrow S-1$ \\
        \bottomrule
    \end{tabular}
\end{table}
\begin{algorithm}[H]
    \caption{Gillespie algorithm for the BSD model}
    \begin{algorithmic}[1]
        \State  \textbf{Input:} Initial state $[H_0, R_0, E_0, I_0, A_0, S_0]$, parameters, $t_{\text{end}}$, $N_{\max}$.
        \State Initialize $t \gets 0$, set initial states and storage arrays.
        \While{$t < t_{\text{end}}$ and step $\leq N_{\max}$}
        \State Compute all propensities $a_1, ..., a_{15}$ for all events
        \State Compute total propensity: $a_0\gets\sum_{i=1}^{15}a_i$
        \If{$a_0=0$}
            \State\textbf{break}
        \EndIf
        \State Generate two random numbers $r_1,r_2\sim \mathcal{U}(0,1)$
        \State Compute time to next event: $\tau\gets-\frac{\ln(r_1)}{a_0}$
        \State Determine event $j$ such that: $\sum_{i=1}^{j-1}a_i<a_0r_2\leq\sum_{i=1}^{j}a_i$
        \State Update state variables and time $t\gets t+\tau$
        \State Repeat until $t>t_{\text{end}}$
        \EndWhile
        \State \textbf{Output:} Trajectories of $H(t), R(t), E(t), I(t), A(t), S(t)$ over time
    \end{algorithmic}
\end{algorithm}

\subsection{Numerical simulations}\label{sec:sim}
We present numerical simulations of deterministic and stochastic BSD models to illustrate disease dynamics, variability, and the probability of disease eradication. Deterministic solutions are obtained from the BSD model \eqref{eqn1}, while the stochastic realizations are generated using Gillespie SSA.

Figures \ref{fig13} and \ref{fig14} compare the temporal evolution of leaf and spore biomass under the deterministic and stochastic frameworks, respectively. Stochastic simulations capture variability through individual realizations; the mean across 400 realizations is shown (red curves), together with the corresponding standard deviation. Deterministic trajectories are shown in black. 
The stochastic paths (blue) exhibit fluctuations that are particularly important during the early stages and at low population levels, reflecting the role of demographic stochasticity. In the leaf compartments (Figure \ref{fig13}), stochastic effects are most pronounced in the dynamics of susceptible and infected biomass, where variability is highest during transient phases. For the pathogen compartments (Figure \ref{fig14}), asexual spores show moderate variability, whereas sexual spores display highly intermittent, burst-like behavior due to mate limitation, a feature not captured by the deterministic model. 

The probability of eradication of the disease is approximated in Figure \ref{fig15}. The plot was obtained by simulating the system across a range of initial host biomass $H(0)+R(0)=5,10,\cdots,10000$. For each population size, 500 independent stochastic realizations were generated over a fixed time horizon of $t_{\max}=100$ days. The initial host composition was chosen such that $H(0)=0.6(H(0)+R(0))$ and $R(0)=0.4(H(0)+R(0))$, maintaining a consistent proportion of susceptible and resistant leaves across all simulations, while small initial pathogen populations were fixed at $E(0)=1$, $I(0)=1$, $A(0)=2$, and  $S(0)=2$. A realization is classified as eradicated if the infected class $I(t)$ reaches zero before the time horizon is reached. The probability of eradication shown in Figure \ref{fig15} is computed as the fraction of realizations that lead to eradication across the 500 trials for each initial population size. Each red point represents the estimated probability, while the solid black curve represents a smoothed trend of these values. These results demonstrate that stochastic effects can substantially influence BSD dynamics, particularly in small populations, where they may alter persistence and eradication outcomes relative to deterministic predictions.

For the mean eradication time analysis of the Black Sigatoka Disease in Figure \ref{fig16}, 500 independent realizations of the model were simulated up to a fixed time horizon of $t_{\max}=100$ days. In each realization, the time to disease eradication was recorded when all infection-related compartments vanished. The mean eradication time was then computed across all realizations, providing a robust measure of eradication dynamics that accounts for both early eradication events and long-term persistence. Figure \ref{fig16} shows that when the initial host biomass $H(0)+R(0)$ is small, the disease tends to die out relatively quickly. However, as the biomass increases, eradication becomes progressively slower, with the mean eradication time approaching the maximum simulation time $t_{\max}=100$ days. This reflects the biological intuition that larger host populations provide more opportunities for pathogen persistence, making eradication less likely and more delayed.

These results also have implications for post-sanitation monitoring. In particular, the sharp decline in eradication probability and the increase in mean eradication time with host biomass suggest the existence of operational host-density thresholds below which follow-up sanitation and surveillance are most likely to succeed. Thus, maintaining plantations below critical post-intervention host biomass levels may improve the probability of pathogen fade-out and reduce the monitoring period required to confirm eradication. 

\begin{figure}[!ht]
    \centering
    \includegraphics[width=1.0\linewidth]{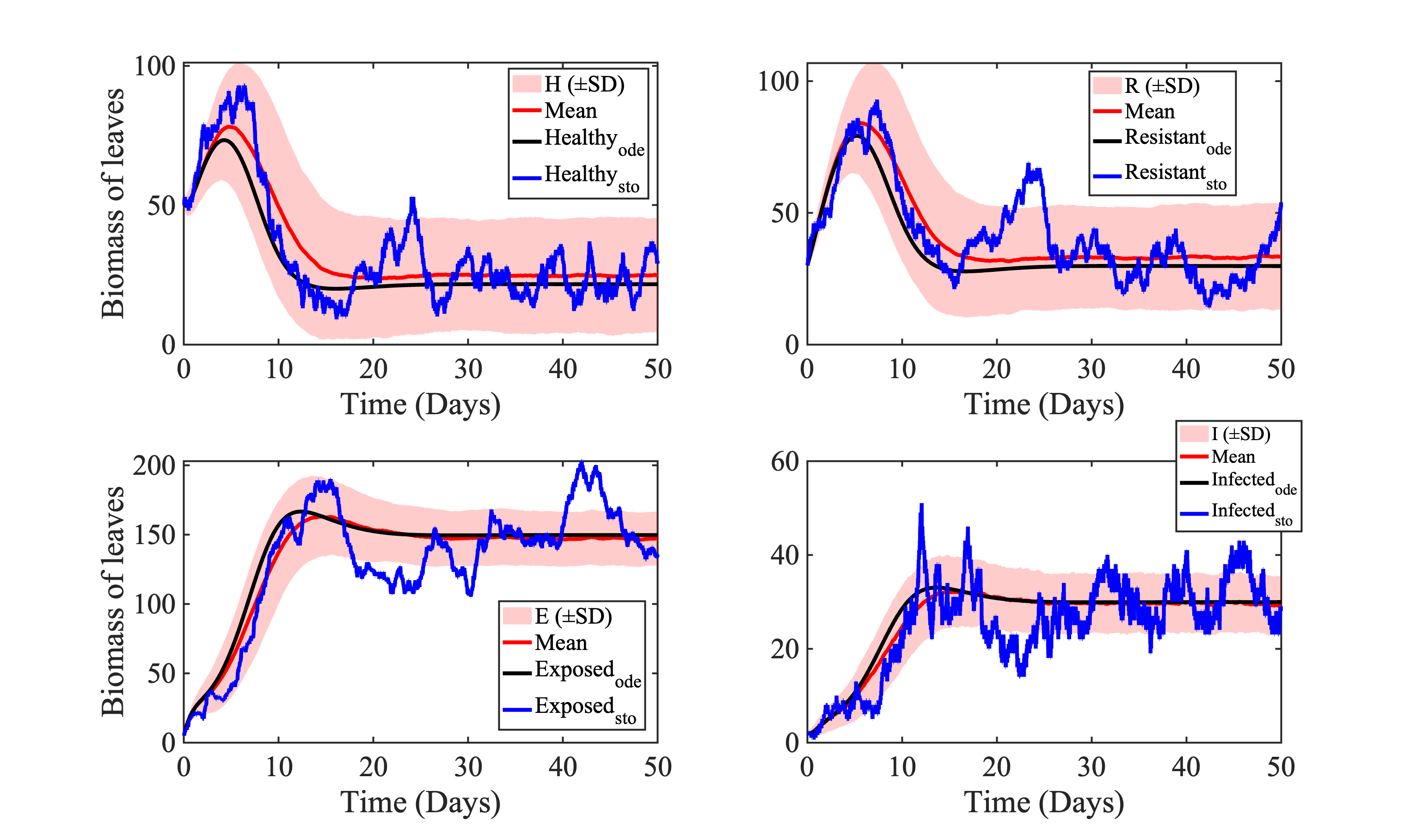}
    \caption{Comparison of deterministic and stochastic simulations of the BSD model over 50 days. Each panel shows the biomass dynamics of a leaf compartment: susceptible leaves $H(t)$, resistant leaves $R(t)$, exposed leaves $E(t)$, and infected leaves $I(t)$. The black curves represent the deterministic trajectories; the blue curves are a single realization of the stochastic simulation; and the red lines indicate the mean across 400 stochastic realizations. Shaded regions show ±1 standard deviation (SD) from the mean of the stochastic simulations. The stochastic system exhibits larger fluctuations, especially in the early dynamics of $H(t)$, $R(t)$, and $I(t)$, capturing variability not seen in the deterministic counterpart.}
    \label{fig13}
\end{figure}

\begin{figure}[!ht]
    \centering
    \includegraphics[width=1.0\linewidth]{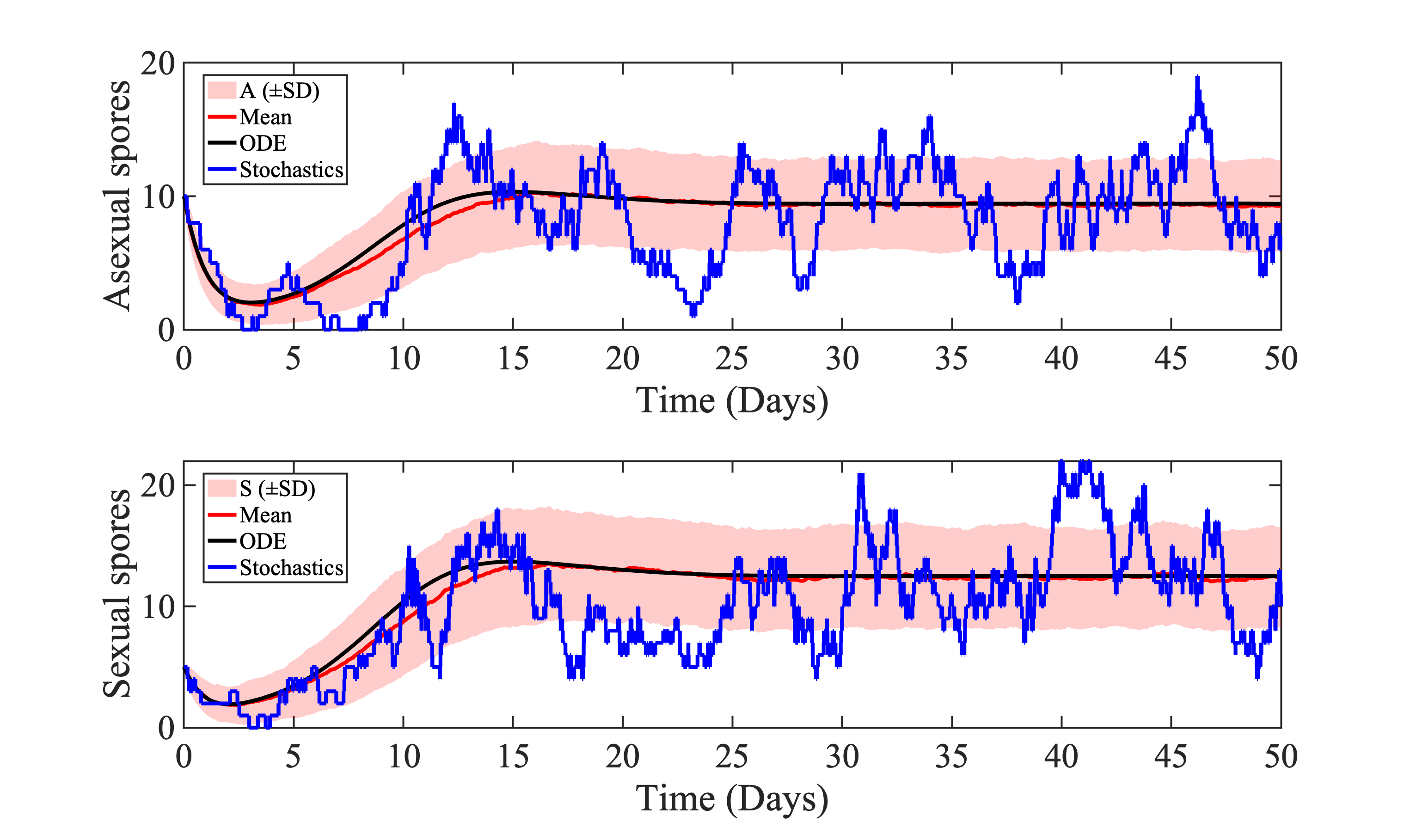}
    \caption{Temporal dynamics of pathogen spores comparing deterministic and stochastic simulations. The top panel shows the biomass of asexual spores $A(t)$, and the bottom panel displays the biomass of sexual spores $S(t)$ over 50 days. The asexual spores exhibit moderate stochastic fluctuations, while the sexual spores show highly discrete and burst-like dynamics due to mating constraints at low densities, which are not captured in the smooth ODE solutions.}
    \label{fig14}
\end{figure}

\begin{figure}[!ht]
    \centering
    \includegraphics[width=1.0\linewidth]{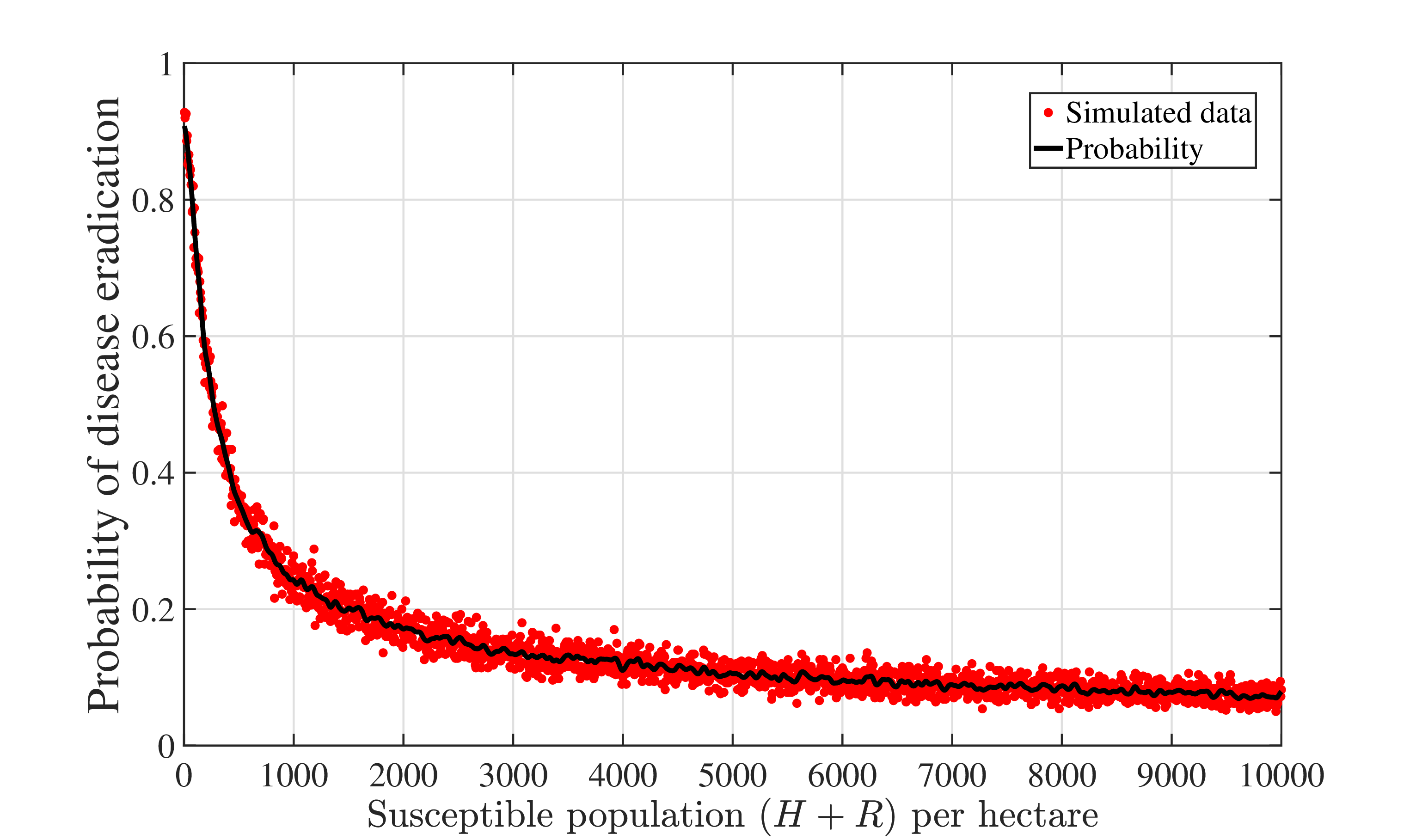}
    \caption{Eradication probability in the model of Black Sigatoka disease as a function of the total initial susceptible and resistant leaf biomass ($H(0) + R(0)$, in leaves per hectare), estimated from 500 stochastic simulations for each initial biomass level. The probability of disease eradication is high when the initial leaf biomass is small (particularly below approximately 1000 leaves per hectare). As the leaf biomass increases, the eradication probability decreases rapidly and stabilizes at a low value, indicating that larger host densities support sustained disease persistence.}
    \label{fig15}
\end{figure}

\begin{figure}[!ht]
    \centering
    \includegraphics[width=1.0\linewidth]{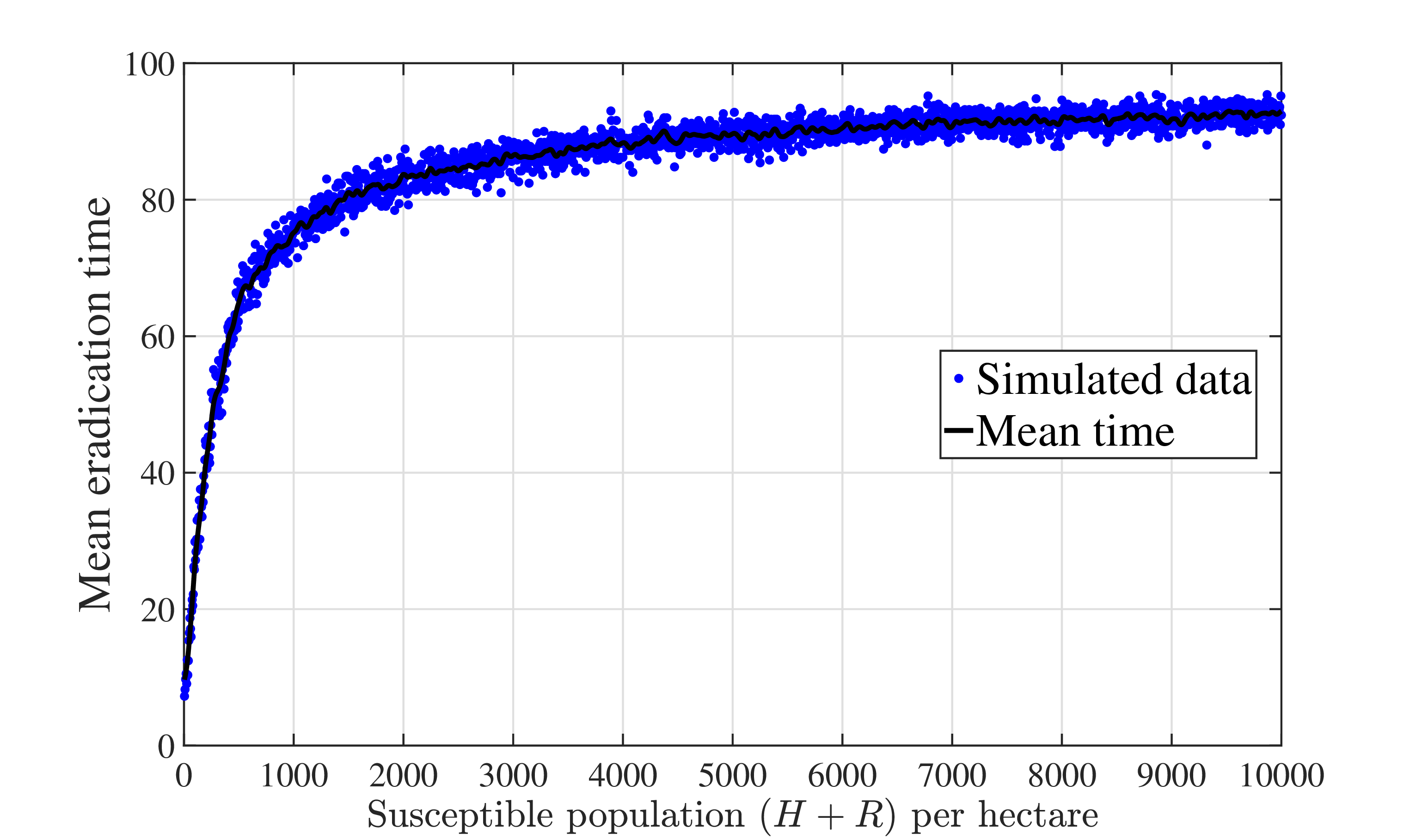}
    \caption{Mean time to disease eradication in the model of Black Sigatoka disease as a function of the total initial susceptible and resistant leaf biomass ($H(0) + R(0)$, in leaves per hectare), estimated from 500 stochastic simulations for each initial biomass level. As the initial leaf biomass increases, the mean eradication time lengthens and approaches the maximum simulation time of 100 days.}
    \label{fig16}
\end{figure}

\clearpage
\subsection{Sobol sensitivity analysis of variability of \texorpdfstring{$I^*$}{}}
To quantify how parameter uncertainty influences the endemic-infected leaf biomass $I^*$, we perform a global sensitivity analysis using Sobol's variance-based method \cite{saltelli, sobol2, sobol}. This approach decomposes the variance of the model output into contributions from individual parameters and their interactions, without assuming monotonic relationships. Parameters are sampled over a prescribed range using Sobol sequences, and $I^*$ is computed for each sample to estimate $\mathrm{Var}(I^*)$. The first-order Sobol index $(F_i)$ measures the direct contribution of a parameter to output variability,  while the total-order index ($T_i$) accounts for both direct effects and all interaction effects.

Figure \ref{sobol} shows that the host recruitment parameters $\kappa, \Lambda$, along with $(\rho)$, sanitation, and cultural practices, are the primary drivers of variability in $I^*$. In particular, $\kappa$ dominates, contributing more than $60\%$ of the variance independently and nearly $70\%$ when interactions are included. Most other parameters exhibit negligible first-order effects, indicating limited individual influence on $I^*$. Figure \ref{sobol_interact} presents the interaction contributions $T_i-F_i$, showing that the variability of $I^*$ is influenced by distributed nonlinear interactions rather than a single interaction pathway. Although $\kappa$, $\Lambda$, and $\rho$ remain important, interaction effects are shared across multiple parameters, indicating that the variability of $I^*$ is driven by host recruitment and sanitation processes, with additional contributions from widespread nonlinear interactions. These results suggest that interventions targeting these key parameters are likely to produce robust reductions in endemic infection levels.
\begin{figure}[!htbp]
    \begin{subfigure}{0.55\textwidth}
	\centering
	\includegraphics[width=1.0\linewidth]{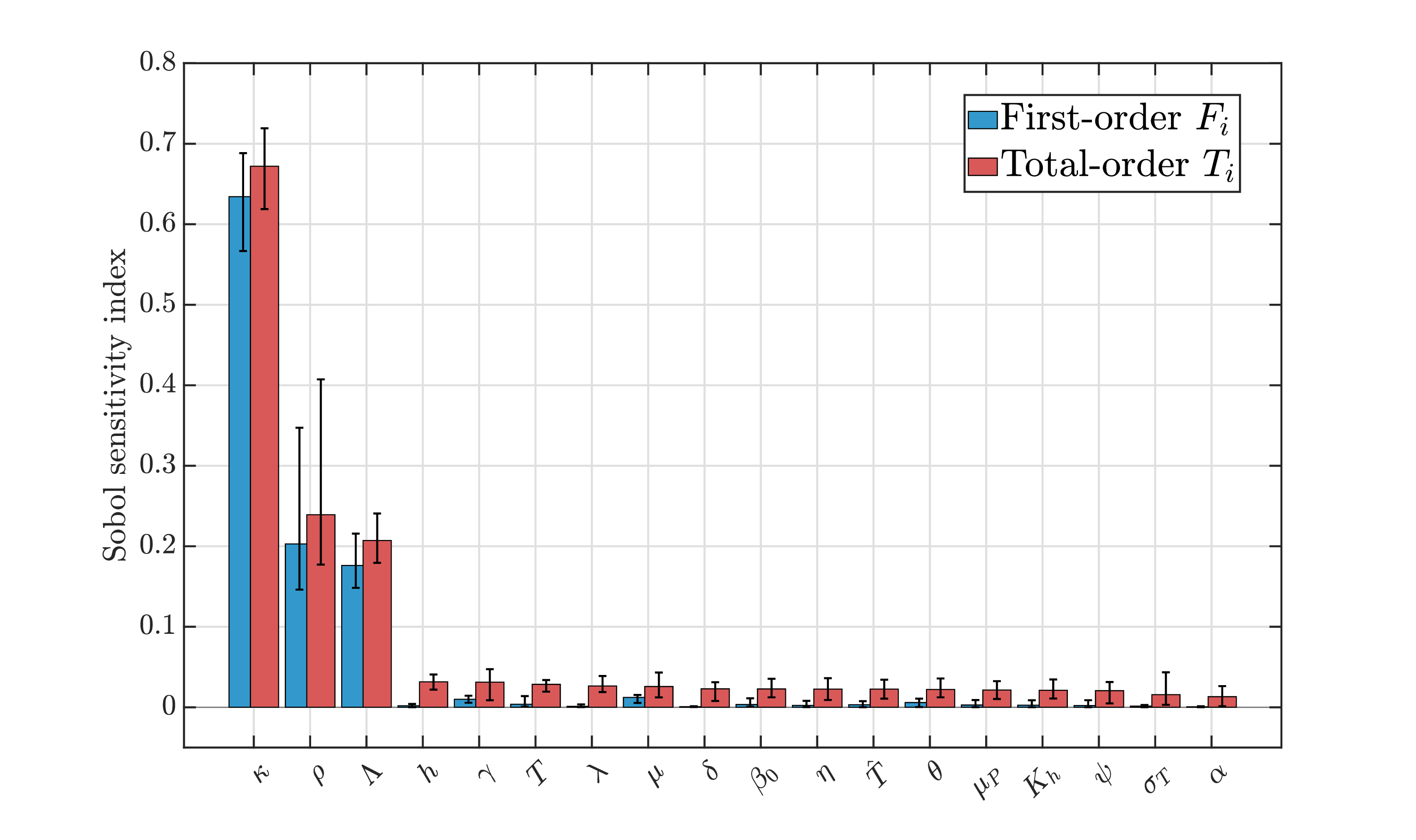}
        \caption{}
        \label{sobol}
    \end{subfigure}
    \begin{subfigure}{0.55\textwidth}
	\centering
	\includegraphics[width=1.0\linewidth]{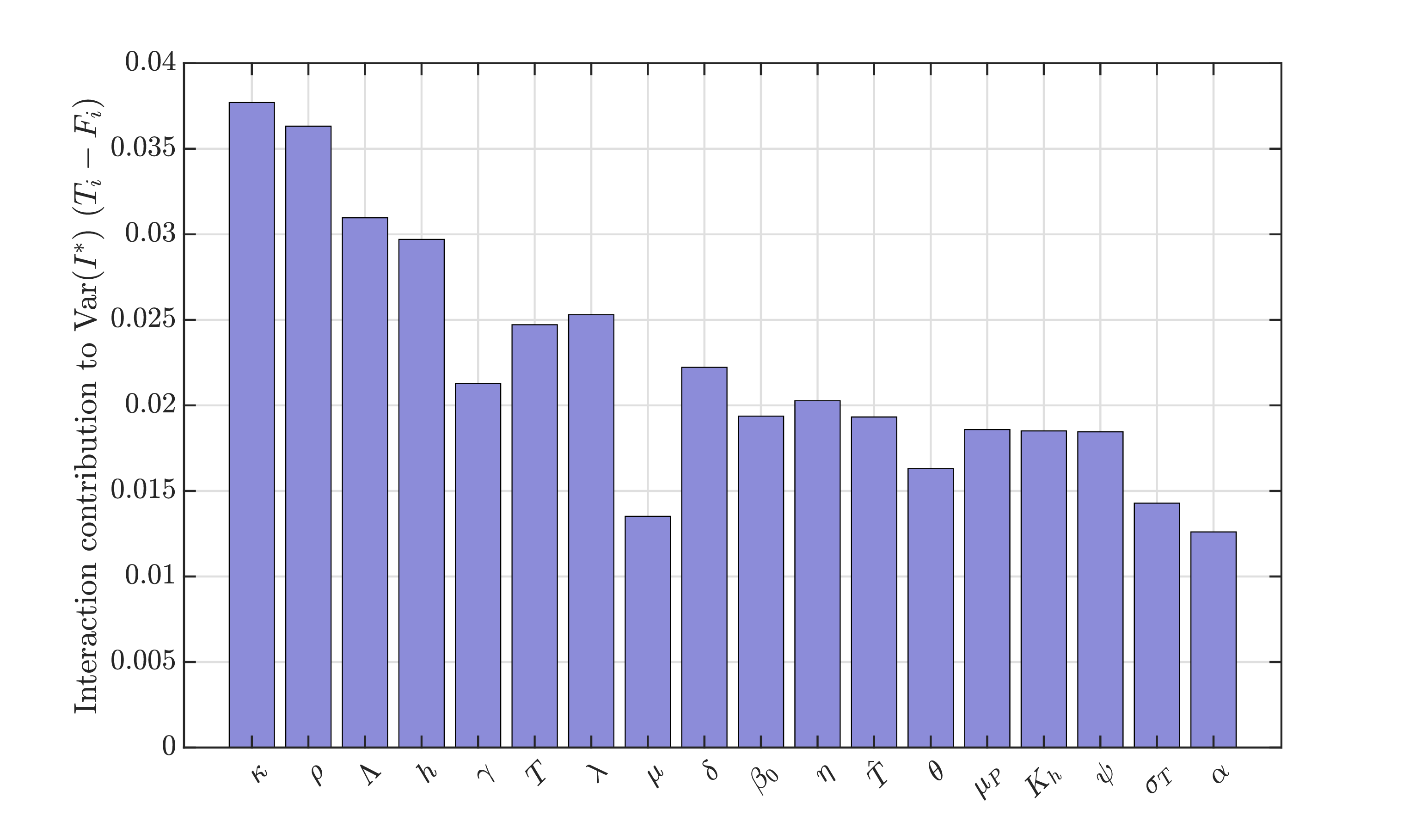}
        \caption{}
        \label{sobol_interact}
    \end{subfigure}
    \caption{Variance-based global sensitivity analysis of the endemic-infected leaf biomass $I^*$ using Sobol indices. \textbf{(a)} shows the first-order $(S_i)$ and total-order $(S_{T_i})$ sensitivity indices, quantifying the direct and total contributions of each parameter to $\mathrm{Var}(I^*)$, respectively. \textbf{(b)} shows the interaction contribution $S_{T_i} - S_i$, highlighting the extent of nonlinear interactions among parameters. Generated using the parameter values stated in Table \ref{table2}.}
    \label{fig4}
\end{figure}

\clearpage
\section{Conclusion}\label{sec:conclusion}
In this study, a deterministic compartmental model incorporating dual transmission pathways (via ascospores and conidia) and mate limitation in sexual reproduction was developed to describe the dynamics of Black Sigatoka disease (BSD) caused by \textit{Pseudocercospora fijiensis}. 
The model exhibits a backward bifurcation, by which a stable endemic equilibrium coexists with the disease-free equilibrium for a range of parameter values where the basic reproduction number satisfies $\mathcal{R}_0<1$. This bistability provides a theoretical explanation for the observed resilience of BSD. Management strategies that focus on reducing $\mathcal{R}_0$ below 1 may prove insufficient in the backward bifurcation regime, as the system can remain trapped in the endemic state despite a subcritical transmission potential.

Sensitivity analysis of the endemic equilibrium, conducted using normalized forward sensitivity indices, Latin Hypercube Sampling, and Partial Rank Correlation Coefficients, identified the most influential parameters. These results emphasize that effective long-term suppression requires interventions extending beyond $\mathcal{R}_0$ reduction. Priority should be assigned to (1) cultural practices that limit the production of new susceptible leaves during periods of elevated transmission risk, e.g., through optimized planting density, timing, or defoliation strategies, and (2) the development and widespread deployment of genetically resistant or tolerant banana varieties, which directly reduce host susceptibility and therefore weaken both transmission pathways. From a management perspective, these results in Figures \ref{nfsi}, 
\ref{fig3}, and \ref{fig4} indicate that interventions aimed at reducing the fraction of newly susceptible leaves $(\kappa)$, lowering host recruitment $(\Lambda)$, and strengthening sanitation/cultural practices $(\rho)$ are likely to produce the largest reductions in endemic infected biomass, especially during periods of elevated humidity and near-optimal temperature.

To account for intrinsic stochastic fluctuations inherent in finite populations and discrete infection events, a stochastic model was formulated using the Stochastic Simulation Algorithm (SSA). Extensive numerical realizations demonstrated that stochastic trajectories generally align with deterministic predictions under moderate to high inoculum pressure. However, at low population sizes, stochastic effects introduce substantial variability, including elevated eradication probabilities even when the deterministic model predicts persistence. Sobol variance-based analysis further demonstrates that variability in endemic-infected leaf biomass is primarily driven by a small set of parameters, with additional contributions from nonlinear interactions. These findings highlight the importance of considering demographic stochasticity when evaluating eradication feasibility, particularly in small or recently sanitized plantations.

Overall, our analyses suggest that sustainable management of Black Sigatoka disease requires an integrated strategy that includes cultural practices and genetic resistance. The backward bifurcation phenomenon, due to a nonlinear mate-limitation term, explains the limited success of transmissibility-focused strategies observed in production regions. Future modeling efforts should address the limitations in our model by incorporating spatial heterogeneity, explicit fungicide-resistance dynamics, or climate-change projections to further refine control recommendations. Prioritizing resistant cultivars and strategically reducing susceptible host tissue offer the most promising path toward reducing reliance on intensive fungicide applications while safeguarding banana and plantain production.

\clearpage
\bibliographystyle{plain}
\bibliography{REFS}

@article{bebber,
  title={Climate change effects on Black Sigatoka disease of banana},
  author={Bebber, Daniel P},
  journal={Philosophical Transactions of the Royal Society B},
  volume={374},
  number={1775},
  pages={20180269},
  year={2019},
  publisher={The Royal Society}
}

@article{agouanet,
  title={Control Model of Banana Black Sigatoka Disease with Seasonality},
  author={Agouanet, Franklin Platini and Yatat-Djeumen, Valaire and Tankam-Chedjou, Isra{\"e}l and Tewa, Jean Jules},
  journal={Differential Equations and Dynamical Systems},
  pages={1--40},
  year={2024},
  publisher={Springer}
}

@article{ravigne,
  title={Mate limitation in fungal plant parasites can lead to cyclic epidemics in perennial host populations},
  author={Ravign{\'e}, Virginie and Lemesle, Val{\'e}rie and Walter, Alicia and Mailleret, Ludovic and Hamelin, Fr{\'e}d{\'e}ric M},
  journal={Bulletin of Mathematical Biology},
  volume={79},
  pages={430--447},
  year={2017},
  publisher={Springer}
}

@article{diekmann,
  title={On the definition and the computation of the basic reproduction ratio R 0 in models for infectious diseases in heterogeneous populations},
  author={Diekmann, Odo and Heesterbeek, Johan Andre Peter and Metz, Johan Anton Jacob},
  journal={Journal of Mathematical Biology},
  volume={28},
  pages={365--382},
  year={1990},
  publisher={Springer}
}

@article{van,
  title={Reproduction numbers and sub-threshold endemic equilibria for compartmental models of disease transmission},
  author={Van den Driessche, Pauline and Watmough, James},
  journal={Mathematical biosciences},
  volume={180},
  number={1-2},
  pages={29--48},
  year={2002},
  publisher={Elsevier}
}

@article{chitnis,
  title={Determining important parameters in the spread of malaria through the sensitivity analysis of a mathematical model},
  author={Chitnis, Nakul and Hyman, James M and Cushing, Jim M},
  journal={Bulletin of Mathematical Biology},
  volume={70},
  pages={1272--1296},
  year={2008},
  publisher={Springer}
}

@article{bari,
  title = {An exploration of modeling approaches for capturing seasonal transmission in stochastic epidemic models},
  journal = {Mathematical Biosciences and Engineering},
  volume = {22},
  number = {2},
  pages = {324-354},
  year = {2025},
  issn = {1551-0018},
  doi = {10.3934/mbe.2025013},
  url = {https://www.aimspress.com/article/doi/10.3934/mbe.2025013},
  author = {Mahmudul Bari Hridoy},
}

@article{ramirez,
  title={SIR-SI model with a Gaussian transmission rate: Understanding the dynamics of dengue outbreaks in Lima, Peru},
  author={Ram{\'\i}rez-Soto, Max Carlos and Machuca, Juan Vicente Bogado and Stalder, Diego H and Champin, Denisse and M{\'a}rtinez-Fern{\'a}ndez, Maria G and Schaerer, Christian E},
  journal={Plos one},
  volume={18},
  number={4},
  pages={e0284263},
  year={2023},
  publisher={Public Library of Science, San Francisco, CA, USA}
}

@article{gillespie,
  title={Exact stochastic simulation of coupled chemical reactions},
  author={Gillespie, Daniel T},
  journal={The journal of physical chemistry},
  volume={81},
  number={25},
  pages={2340--2361},
  year={1977},
  publisher={ACS Publications}
}

@article{afful,
  title={Deterministic optimal control compartmental model for COVID-19 infection},
  author={Afful, Bernard Asamoah and Safo, Godfred Agyemang and Marri, Daniel and Okyere, Eric and Ohemeng, Mordecai Opoku and Kessie, Justice Amenyo},
  journal={Modeling Earth Systems and Environment},
  volume={11},
  number={2},
  pages={87},
  year={2025},
  publisher={Springer}
}

@article{henderson,
  title={Black Sigatoka disease: new technologies to strengthen eradication strategies in Australia},
  author={Henderson, Juliane and Pattemore, JA and Porchun, SC and Hayden, HL and Van Brunschot, S and Grice, Kathy RE and Peterson, RA and Thomas-Hall, SR and Aitken, EAB},
  journal={Australasian Plant Pathology},
  volume={35},
  number={2},
  pages={181--193},
  year={2006},
  publisher={Springer}
}

@article{sosnowski,
  title={Techniques for the treatment, removal and disposal of host material during programmes for plant pathogen eradication},
  author={Sosnowski, MR and Fletcher, JD and Daly, AM and Rodoni, BC and Viljanen-Rollinson, SLH},
  journal={Plant Pathology},
  volume={58},
  number={4},
  pages={621--635},
  year={2009},
  publisher={Wiley Online Library}
}

@article{peterson,
  title={Eradication of black leaf streak disease from banana growing areas in Australia},
  author={Peterson, R and Grice, K and Goebel, R},
  journal={InfoMusa},
  volume={14},
  number={2},
  pages={2},
  year={2005}
}

@article{cook,
  title={Predicted economic impact of black Sigatoka on the Australian banana industry},
  author={Cook, David C and Liu, Shuang and Edwards, Jacqueline and Villalta, Oscar N and Aurambout, Jean-Philippe and Kriticos, Darren J and Drenth, Andre and De Barro, Paul J},
  journal={Crop Protection},
  volume={51},
  pages={48--56},
  year={2013},
  publisher={Elsevier}
}

@article{castillo,
  title={Dynamical models of tuberculosis and their applications},
  author={Castillo-Chavez, Carlos and Song, Baojun},
  journal={Mathematical biosciences and engineering},
  volume={1},
  number={2},
  pages={361},
  year={2004},
  publisher={American Institute of Mathematical Sciences}
}

@article{sobol,
  title={On sensitivity estimation for nonlinear mathematical models},
  author={Sobol', Il'ya Meerovich},
  journal={Matematicheskoe modelirovanie},
  volume={2},
  number={1},
  pages={112--118},
  year={1990},
  publisher={Russian Academy of Sciences, Branch of Mathematical Sciences}
}

@article{sobol2,
  title={Global sensitivity indices for nonlinear mathematical models and their Monte Carlo estimates},
  author={Sobol, Ilya M},
  journal={Mathematics and computers in simulation},
  volume={55},
  number={1-3},
  pages={271--280},
  year={2001},
  publisher={Elsevier}
}

@book{saltelli,
  title={Sensitivity analysis in practice: a guide to assessing scientific models},
  author={Saltelli, Andrea and Tarantola, Stefano and Campolongo, Francesca and Ratto, Marco and others},
  volume={1},
  year={2004},
  publisher={Wiley Online Library}
}

\end{document}